\documentclass[a4paper,UKenglish,cleveref, autoref, thm-restate,svgnames,dvipsnames]{lipics-v2021}
\pdfoutput=1

\hideLIPIcs  

\title{On the Complexity of Problems on Graphs Defined on Groups } 

\author{Bireswar Das}{IIT Gandhinagar, India}{bireswar@iitgn.ac.in }{}{}

\author{Dipan Dey}{TIFR, India}{dipan.dey@tifr.res.in}{}{}

\author{Jinia Ghosh}{IIT Gandhinagar, India}{jiniag@iitgn.ac.in}{}{}

\authorrunning{B. Das, D. Dey and J. Ghosh} 

\Copyright{} 
    
\ccsdesc[500]{Theory of computation~Computational complexity and cryptography, Design and analysis of algorithms}

\keywords{Complexity Theory, \and Graphs Defined on Groups,  \and Exponential Time Hypothesis, \and Power Graphs, \and Graph Motif} 

\category{} 

\relatedversion{A preliminary version of this paper has been accepted in the 25th International Symposium on Fundamentals of Computation Theory (FCT 2025).} 
\acknowledgements{We would like to thank V. Arvind and P. J. Cameron for their valuable inputs. We also thank the three anonymous reviewers of FCT 2025 for their insightful feedback and comments.}

\nolinenumbers 

\EventEditors{John Q. Open and Joan R. Access}
\EventNoEds{2}
\EventLongTitle{42nd Conference on Very Important Topics (CVIT 2016)}
\EventShortTitle{CVIT 2016}
\EventAcronym{CVIT}
\EventYear{2016}
\EventDate{December 24--27, 2016}
\EventLocation{Little Whinging, United Kingdom}
\EventLogo{}
\SeriesVolume{42}
\ArticleNo{23}

\usepackage{charter,eulervm,parskip,tikz,float,nicefrac,framed}
\usepackage{xcolor,graphicx}
\usetikzlibrary{shapes.geometric, calc, positioning, decorations.pathreplacing}
\usetikzlibrary{decorations.shapes,decorations.pathreplacing,decorations.pathmorphing}
\usetikzlibrary{arrows,matrix,shapes}
\usepackage{tcolorbox}
\tikzset{
    small circles/.style={circle,inner sep=2pt,fill=#1},
    hollow circles/.style n args={2}{circle,inner sep=#1,draw=#2,thick},
    stars/.style={star,inner sep=2pt}
}
\definecolor{DarkGreen}{rgb}{0.1,0.5,0.1}
\hypersetup{
colorlinks=true,
linkcolor=IndianRed,
urlcolor=DarkGreen,
citecolor=SeaGreen
}

\usepackage{tikz}
\usetikzlibrary{decorations.markings}
\makeatletter
\let\c@author\relax
\makeatother
\usepackage[giveninits=true, maxalphanames=99, minalphanames=1,maxnames=9]{biblatex}
\newtheorem{fact}{Fact}

\theoremstyle{claimstyle}
\newtheorem{claimc}{Claim}
\newtheorem{claimd}{Claim}
\newenvironment{proofof}[1]{ \vspace{0.5 em}
    \par\noindent\textit{Proof of #1.}\;%
}{%
    \hfill\par%
}

\newcommand{\G}{\Gamma}
\newcommand{\PP}{\mathcal{P}}
\newcommand{\GG}{\Gamma_{\Phi}}
\newcommand{\DG}{\mathcal{D}}
\newcommand{\ZZ}{\mathbb{Z}}

\newcommand{\DDG}[1]{\mathcal{D}_{#1}}
\begin{document}
%



\maketitle              
\begin{abstract}

We study the complexity of graph problems on graphs defined on groups, especially power graphs. We observe that an isomorphism invariant problem, such as Hamiltonian Path, Partition into Cliques, Feedback Vertex Set, Subgraph Isomorphism, cannot be NP-complete for power graphs, commuting graphs, enhanced power graphs, directed power graphs, and bounded-degree Cayley graphs, assuming the Exponential Time Hypothesis (ETH). An analogous result holds for isomorphism invariant group problems: no such problem can be NP-complete unless ETH is false. We show that the Weighted Max-Cut problem is NP-complete in power graphs. 
We also show that, unless ETH is false, the Graph Motif problem cannot be solved in quasipolynomial time on power graphs, even for power graphs of cyclic groups. 
We study the recognition problem of power graphs when the adjacency matrix or list is given as input and show that for abelian groups and some classes of nilpotent groups, it is solvable in polynomial time.
\end{abstract}

\section{Introduction}
Several important classes of graphs are defined on groups, including power graphs, enhanced power graphs, Cayley graphs, and commuting graphs. These constructions encode algebraic relationships into combinatorial structures, creating an interplay between group theory and graph theory. In this paper, we study the computational complexity of various problems on such graphs. Since these graphs are derived from algebraic objects, a natural and compelling question is whether the underlying group structure can be leveraged to design more efficient algorithms or to better understand the inherent difficulty of the problems. 



In this paper, while our primary focus is on power graphs,  several of our results also extend to enhanced power graphs, directed power graphs, commuting graphs and bounded degree Cayley graphs (see \cite{cameron2022graphs} for a comprehensive survey). 
The notion of directed power graphs was introduced by Kelarev et al. \cite{kelarev2000combinatorial}, while the concept of power graphs was later proposed by Chakrabarty et al. \cite{chakrabarty2009undirected}. 
As power graphs are perfect graphs \cite{aalipour2017structure}, using results by 
Lovasz et al. 
\cite{gr6tschel1988geometric} for perfect graphs, one can see that computational problems such as graph colouring, maximum clique and maximum independent set are in P.



 While power graphs had been defined decades ago, to the best of our knowledge, no nontrivial NP-completeness result has been reported for this class. Our first line of investigation concerns identifying NP-complete problems on graphs defined from groups. Many classical NP-complete problems,  such as  Hamiltonian Cycle, Partition into Cliques, Partition into Perfect Matchings, Metric Dimension, Treewidth determination, are invariant under graph isomorphism, meaning their solutions depend only on the structure of the graph, not on the labelling of the vertices. We prove that any isomorphism-invariant graph problem is not NP-complete when restricted to power graphs, enhanced power graphs, commuting graphs, directed power graphs, or bounded-degree Cayley graphs, unless Exponential Time Hypothesis (ETH) fails (\Cref{sec : non-hardness}). This result also extends to isomorphism invariant group problems where the groups are given by their Cayley table. It is already known that isomorphism invariant problems such as the group isomorphism problem (see \cite{miller1978nlog}), 
minimal faithful permutation degree problem \cite{das2024minimal} are not NP-complete under ETH. Our result shows that \emph{any} isomorphism invariant group problem cannot be NP-complete, under ETH.

Therefore, to identify NP-complete problems on group-based graphs, we must focus on problems that are not invariant under graph isomorphism (\Cref{sec: motif}). One such example is the Weighted Max-Cut problem. We show that Weighted Max-Cut remains NP-complete when restricted to power graphs. However, this result is not particularly surprising, as for graph classes having sufficiently large cliques, the NP-completeness of Weighted Max-Cut is expected. Motivated by this, we turn to the Graph Motif problem, which differs fundamentally in that it incorporates a connectivity constraint in addition to a colouring condition. 


Originally introduced in the context of biological network analysis \cite{lacroix2006motif}, the graph motif problem has since attracted significant interest within the computer science community. The problem remains NP-complete even when restricted to trees \cite{lacroix2006motif} and bipartite graphs with maximum degree $4$ \cite{fellows2007sharp}. A variety of fixed-parameter tractable (FPT) algorithms have been developed for this problem, with parameters ranging from the number of distinct colours to the input size (see \cite{pinter2016deterministic, guillemot2013finding}, e.g.). 
Ganian showed that Graph Motif can be solved in time $O^*(2^k)$ where $k$ is the neighbourhood diversity \cite{ganian2012using}. Notably, the problem is $W[1]$-hard for trees when parameterised by the number of colours \cite{fellows2007sharp}. Furthermore, several lower bounds based on ETH have been established for the problem \cite{bonnet2017graph}.


We prove that the Graph Motif problem does not admit a quasipolynomial time algorithm on power graphs, even for cyclic groups, assuming ETH (\Cref{sec: motif}).  
This shows that even with the strong algebraic restrictions of a cyclic group, the motif problem remains hard. This hardness result is derived through a careful application of concepts from both group theory and number theory.

On the contrary, by exploiting 
the structure of power graphs of cyclic groups, we obtain a subexponential time algorithm to solve the Graph Motif problem for power graphs of cyclic groups. Moreover, for power graphs of $p$-groups, the problem admits a more efficient solution, as we present a polynomial time algorithm.

On a related note, it is worthwhile to mention that the non-hardness results for isomorphism-invariant graph (or group) problems defined on graphs (or groups), as well as the hardness results of the Graph Motif problem in power graphs, can be derived from the widely held assumption that $\mathsf{EXP} \neq \mathsf{NEXP}$, where $\mathsf{EXP} = \bigcup_{t>0}\mathsf{DTIME}(2^{n^t})$ and $\mathsf{NEXP} = \bigcup_{t>0}\mathsf{NTIME}(2^{n^t})$, by using a result by Buhrman and Homer \cite{10.1007/3-540-56287-7_99}. 

There are two natural representations of power graphs. The first is via the Cayley table of the underlying group, where the group structure is explicitly provided. We refer to this as \emph{the algebraic representation}. The second is through the adjacency matrix or adjacency list of the power graph, with no accompanying group information.  In the latter representation, the graph is presented purely as a combinatorial object. 
We refer to this as the \emph{standard graph representation}.


Within the standard graph representation, a fundamental question is the recognition problem: given a graph, determine whether it is the power graph of some group.  Following the results of Arvind et al. for recognition of commuting graphs \cite{arvind2024aspectscommutinggraph}, one can show that this problem admits a quasipolynomial time algorithm for power graphs, enhanced power graphs, and directed power graphs. We show that it is easy to recognise power graphs of nilpotent groups of squarefree exponent. We further show that the recognition problem for power graphs is solvable in polynomial time when the graph arises from an abelian group, as well as from certain classes of nilpotent groups (\Cref{sec: recognition}). In particular, we consider nilpotent groups of bounded polycyclic length - a class that includes several important subclasses, such as metacyclic groups and nilpotent groups of cube-free order.



This paper presents both hardness and non-hardness results on power graphs, along with recognition algorithms for specific group classes. Except for the result on $p$-groups in \Cref{sec: motif}, all results in \Cref{sec : non-hardness} and \ref{sec: motif} are representation-independent. For $p$-groups, however, the \textsc{Graph Motif} problem must be treated as a promise problem under standard representation or require algebraic input.

\section{Preliminaries}


For a \textit{graph} $\Gamma$, we write $V(\Gamma)$ and $E(\Gamma)$ to denote the vertex set and edge set, respectively.
We denote an \textit{undirected edge} between two vertices $u$ and $v$ by $\{u,v\}$ and a \textit{directed edge} from $u$ to $v$ by $(u,v)$. Unless otherwise specified, the graphs we consider are simple undirected graphs. The \textit{neighbourhood} of a vertex $u$ in a graph $\G$ is denoted by $N_{\G}(u)$. For a subset $L \subseteq V(\G)$, $N(L)=\{v \in V(\G) : \text{ there exists } a \in L \text{ with } \{a,v\} \in E(\G) \}$. 
Two vertices $u$ and $v$ are \textit{closed twins} in a \emph{directed graph} $\DG$ if their closed out-neighbourhoods are the same. 
The relation between two vertices of being closed twins is an equivalence relation, and the equivalence class under this relation is called the \textit{closed twin-class}. 
The \textit{subgraph induced on the set} $S\subseteq V(\G)$ is denoted by $\Gamma[S]$. 
A function $f:V(\G) \xrightarrow[]{} V(\G')$ is an \textit{isomorphism} between two graphs $\G$ and $\G'$ if $f$ is a bijection 
such that $\{u,v\} \in E(\G)$ if and only if $\{f(u),f(v)\} \in E(\G')$; here, $\G$ and $\G'$ are called \textit{isomorphic} and it is denoted by $\G \cong \G'$. We write $col(v)$ to denote the colour of a vertex $v$ in a vertex-coloured graph. For two vertex-coloured graphs $\G$ and $\G'$, if for an isomorphism $f$, 
$col(u)=col(f(u))$, for all $u\in V(\G)$, then $f$ is called a \textit{colour-preserving isomorphism} (or, \emph{colour-isomorphism}) and $\G$ and $\G'$ are said to be \textit{colour-isomorphic}, denoted by $\G \cong_c \G'$. 

The basic definitions and facts on group theory can be found in any standard book (e.g.,  \cite{rotman2012introduction}). 
The \textit{order of a group} $G$ and the \textit{order of an element} $g$ in $G$ are denoted by $|G|$  and $ord(g)$ respectively. The \textit{cyclic subgroup generated by $g$} is denoted by $\langle g\rangle$, where $g$ is a \textit{generator} of $\langle g \rangle$. The number 
of generators of $\langle g \rangle$ is $\phi(|\langle g \rangle|)$, where $\phi$ is the \textit{Euler's totient function} and $\phi(p_1^{\alpha_1}\dots p_k^{\alpha_k})=p_1^{\alpha_1-1}(p_1-1)\dots p_k^{\alpha_k-1}(p_k-1)$, where $p_i$'s are distinct primes. Any cyclic group of order $n$ is isomorphic to $\ZZ_n$, the additive group of integers modulo $n$. The \textit{exponent} of a finite group is the least common multiple of the orders of all elements of the group.



A group $G$ is called a \textit{$p$-group} if the order of $G$ is a power of a prime $p$. 
It is common to refer to a group as a $p$-group whenever its order is a power of some prime, where the symbol $p$ denotes an arbitrary prime. 
If $p^m$ is the highest power of a prime $p$ such that $p^m$ divides $|G|$, then there is at least one subgroup of order $p^m$ in $G$, and such a subgroup is called \textit{Sylow p-subgroup} of $G$. A finite group is called a \textit{nilpotent group} if it is a direct product of its Sylow subgroups. Moreover, each Sylow subgroup is unique in a finite nilpotent group.

Since our primary focus is on power graphs and directed power graphs, we provide their definitions only. For the definitions of other group-based graph classes, 
we refer the reader to \cite{cameron2022graphs}.
The \emph{directed power graph} of a group $G$, denoted by $DPow(G)$, is the directed graph with vertex set $G$, and edge set $\{(x,y): y = x^m$ \textrm{ for some positive integer} m\}. The \emph{power graph}, denoted by $Pow(G)$, is the underlying undirected graph of $DPow(G)$, i.e., the undirected graph obtained from $DPow(G)$ after removing directions, self-loops and double arcs.

\medskip 

 \begin{remark}\label{remark power graph} 
 If $\{x,y\}$ $\in E(Pow(G))$, then either $ord(x)|ord(y)$ or $ord(y)|ord(x)$. If $G$ is a cyclic group, then 
 $\{x,y\}$ is an edge in $Pow(G)$ if and only if either $ord(x)|ord(y)$ or $ord(y)|ord(x)$.
 \end{remark}
 
 We define an equivalence relation $\sim$ on $G$ as follows: for $x,y\in G$, $x\sim y$ if and only if $\langle x \rangle = \langle y \rangle$, i.e., $x$ and $y$ generate the same cyclic subgroup of $G$. The equivalence class containing $x$ is denoted by $[x]$. So, all the elements of a class have the same order. Note that the size of $[x]$ is $\phi(ord(x))$.

\medskip 

\begin{remark}\label{equivalence classes form a clique} In $Pow(G)$, the following two facts hold:
    (\textit{$i$}) Each equivalence class $[x]$ of $G$ is a clique of size $\phi(ord(x))$ in $Pow(G)$;
    (\textit{$ii$}) For two classes $[x]$ and $[y]$, if an element $x\in [x]$ is adjacent to an element $y\in [y]$ in $Pow(G)$, then every element of $[x]$ is adjacent to every element of $[y]$. 
\end{remark}
\medskip 
\begin{lemma}\cite{chakrabarty2009undirected}\label{complete power}
    The power graph of a finite group is complete if and only if the group is a cyclic group of prime-power order.
\end{lemma}



 The \textit{Exponential Time Hypothesis} (ETH) is a conjecture by Impagliazzo, Paturi and Zane \cite{impagliazzo2001problems}
 that states that there is no $2^{o(n)}$-time algorithm for $3$-SAT over $n$ variables. 
 
 A function $f$ that takes a graph as input and outputs a graph is called a \emph{canonization function} if a) $f(G)\cong G$ and b) $G_1\cong G_2$ if and only if $f(G_1)= f(G_2)$. 
 
 The size of a set $S$ is denoted by $|S|$, whereas the length of an object $x$ - that is, the number of bits in its encoding - is denoted by $||x||$.

\section{Non-hardness of Isomorphism-Invariant Problems}\label{sec : non-hardness}

Let $\mathscr{G}$ be a class of graphs closed under isomorphism. We consider problems whose instances are of the form $(\Gamma_1,\ldots, \Gamma_k,x)$, where $\Gamma_i\in \mathscr{G}$, $k$ is a constant,  and  $||x||=O(\log s)$, where $s=||(\Gamma_1,\ldots, \Gamma_k,x)||$, is the size of the instance. In other words, a language representation of such problems is a subset of $\mathscr{G}^k\times \{0,1\}^*$, where $k$ is a constant, and the last component (i.e., $x$) is of logarithmic length in the length of the input. We call such a language $L$ \emph{isomorphism invariant over} $\mathscr{G}$, if for all $\Lambda_i,\Gamma_i\in \mathscr{G}$, where $1\leq i \leq k$ 
such that $\Gamma_i\cong \Lambda_i$, 
for all $i$, we have $(\Gamma_1,\ldots,\Gamma_k,s)\in L$ if and only if $(\Lambda_1,\ldots,\Lambda_k,s)\in L$. A wide range of problems are isomorphism invariant. For example, Graph Colouring: The language is $L=\{(\Gamma,r)|\Gamma \textrm{ has a vertex colouring with } r\textrm{ colours}\}$, where $r$ is given in binary and 
$||r||=O( \log |V(\Gamma)|)$. This problem is isomorphism invariant because $\Gamma$ is $r$-colourable if and only if an isomorphic copy of $\Gamma$ is $r$-colourable. 
Similarly, the following problems are isomorphism invariant: Independent Set, Clique, Graph Colouring, Hamiltonian Path, Partition into Cliques, Oriented Diameter, etc. 


The following theorem is stated for isomorphism invariant problems on power graphs. However, the theorem also holds for directed power graphs, enhanced power graphs, commuting graphs, and bounded-degree Cayley graphs.

\medskip 

\begin{theorem}\label{iso-inv}
    If $L$ is an isomorphism invariant graph problem over the class of power graphs, then $L$ is not NP-complete under ETH.
\end{theorem}

We prove the above theorem using a version of Mahaney's theorem \cite{mahaney1982sparse}, which can be proved easily by mimicking the proof of Mahaney's theorem (see, e.g., \cite{grochow2016np}). We say that a language $L$ is \emph{quasipolynomially sparse} if there exists a constant $c>0$ such that $|L\cap \{0,1\}^n|=2^{O(\log^c n)}$. 

\medskip 

\begin{theorem}\label{maha}
Under ETH, no quasipolynomially sparse language is NP-complete.
\end{theorem}

\medskip 

\begin{remark}  

    An isomorphism invariant problem on power graphs need not be quasipolynomially sparse,  irrespective of the representations (algebraic or standard graph representation), even though the number of non-isomorphic groups of order $n$ is quasipolynomial \cite{mciver1987enumerating}. The following lemma 
    is possibly a folklore and has been used previously  (e.g., \cite{farzan2006succinct}).
\end{remark}

\medskip

\begin{lemma}\label{number of cayley tables}
 The number of distinct Cayley tables representing groups of order $n$ is exponential in $n$.
\end{lemma}
\begin{proof}
    Let $C$ be the Cayley table of a group $G$ on elements $\{1,\ldots,n\}$, i.e.,  $C[i,j]$ stores the product of  $i$ and $j$. By relabelling the elements of $G$ by a permutation $\pi$ of $[n]$, we get the Cayley table, say $C^\pi$, of an isomorphic copy of $G$. We have $C^\pi[i,j]=C[ i^{\pi^{-1}},  j^{\pi^{-1}} ]$. One can check that $C$ and $C^\pi$ are the same if and only if $\pi$ is a group automorphism of the group $G$. By the orbit-stabiliser theorem, we obtain that the number of different Cayley tables representing isomorphic copies of $G$ is $n!/|Aut(G)|$, where $Aut(G)$ is the set of group automorphisms of $G$. If we take $G=\mathbb{Z}_n$, then $|Aut(G)|=\phi(n)\leq n$, where $\phi$ is the Euler totient function. Thus, the number of different Cayley tables is at least $(n-1)!=2^{\Omega(n\log n)}$. 
 \end{proof}


The number of adjacency matrices representing power graphs on $n$ vertices can be very small: If $n$ is a prime power,  then the power graph is a complete graph (\Cref{complete power}). However, this is the only case when this number is small, as shown in the next lemma. 

\medskip 

\begin{lemma}\label{number of adj matrices}
    The number of distinct adjacency matrices representing power graphs on $n$ vertices is $\Omega(\exp(\sqrt{n}))$, when $n$ is a number with at least two distinct prime divisors. For infinitely many $n$, there are $\Omega(\exp(n/\log \log n))$ distinct adjacency matrices representing power graphs on $n$ vertices.
\end{lemma}
\begin{proof}
    We show that the number of distinct adjacency matrices representing isomorphic copies of power graphs of $\mathbb{Z}_n$ is itself $\Omega(\exp(\sqrt{n}))$. Similar to the case above for Cayley tables (\Cref{number of cayley tables}), the number of distinct adjacency matrices representing isomorphic copies of a graph $X$ is $n!/|Aut(X)|$, where $Aut(X)$ is the set of automorphisms of the graph $X$. It is known that $|Aut(Pow(\mathbb{Z}_n))|=(\phi(n)+1)!\prod_{d|n, d\notin \{1,n\}} (\phi(d))!$ \cite{FENG2016197}. Now, note the following:
    \begin{equation*}
    \begin{split}
         \frac{n!}{(\phi(n)+1)! \prod_{d|n , d\notin\{1,n\} } (\phi(d))! } & \geq \frac{n!}{(\phi(n)+1)!  (\sum_{d|n , d\notin\{1,n\} } \phi(d))! }\\
         & = \frac{n!}{(\phi(n)+1)!(n-\phi(n)-1)!}\\
         & = \binom {n}{\phi(n)+1}.
    \end{split}   
    \end{equation*}
    The first equality holds because $\sum_{d |n} \phi(d)=n$. From the fact $\sqrt{n}/2 \leq \phi(n) \leq  n - \sqrt{n}$ (the second inequality holds because $n$ has at least two distinct prime divisors), we have $\binom {n}{\phi(n)+1} \geq \binom{n}{\sqrt{n}}$. Using Stirling's approximation, 
    we get $\binom{n}{\sqrt{n}}=\Omega(exp(\sqrt{n}))$.
    
    For infinitely many $n$, it is known that $\frac{n}{c_1\log\log n} < \phi(n) <  \frac{n}{c_2\log\log n}$, where $c_1$ and $c_2$ are two fixed positive constants (Note that such an $n$ cannot be a prime power and therefore, the result from \cite{FENG2016197} can be applied.). The first inequality holds for any $n\geq 3$ (see \cite{rosser1962approximate}), and the second inequality holds for infinitely many $n$ (see 
    \cite{ribenboim2012new}). Using Stirling's approximation, we get $\binom{n}{\phi(n)+1}=\Omega(exp(n/ \log\log n))$.
    \end{proof}

{\bf Proof of \Cref{iso-inv}:} Let $(\Gamma_1\ldots,\Gamma_k,x)$ be an instance of problem $L$. We use Babai's quasipolynomial time computable graph canonization function \cite{babai2016graph}, which we denote by $canon$, on each of $\Gamma_1,\ldots,\Gamma_k$ to obtain $canon(\Gamma_1),\ldots,$ $canon(\Gamma_k)$. Since $L$ is isomorphism invariant, $(\Gamma_1,\ldots,\Gamma_k,x)\in L $ if and only if $(canon(\Gamma_1),\ldots, canon(\Gamma_k),x)\in L$.

Let $L'=\{(canon(\Gamma_1),\ldots, canon(\Gamma_k),x)| (\Gamma_1,\ldots,\Gamma_k,x)\in L\}$. Since the number of non-isomorphic groups of order $n$ is $2^{O(\log^3 n)}$  \cite{mciver1987enumerating}, the number of possibilities of $canon(\Gamma_i)$ is at most quasipolynomial in the size of the input. Since $k$ is a constant, we observe that $L'$ is a quasipolynomially sparse language. Therefore, $L'$ is not NP-hard under ETH by \Cref{maha}. Since $L$ is quasipolynomial time reducible to $L'$, $L$ is also not NP-hard under ETH. \hfill$\blacktriangleleft$


A similar notion of isomorphism invariant problems on groups can be defined. Let $\mathscr{G}$ be the class of finite groups closed under isomorphism. We consider problems whose instances are of the form $(\Gamma_1,\ldots, \Gamma_k,x)$, where $\Gamma_i\in \mathscr{G}$ is given by its Cayley table, $k$ is a constant,  and  $||x||=O(\log s)$, where $s=||(\Gamma_1,\ldots, \Gamma_k,x)||$, is the size of the instance. 
Such a language $L$ is called \emph{an isomorphism invariant group problem}, if for all $\Lambda_i,\Gamma_i\in \mathscr{G}$ with $\Gamma_i\cong \Lambda_i$, where $1\leq i \leq k$, 
 we have $(\Gamma_1,\ldots,\Gamma_k,s)\in L$ if and only if $(\Lambda_1,\ldots,\Lambda_k,s)\in L$. 
The next result is similar.

\medskip 

\begin{theorem}\label{iso-inv-group}
   If $L$ is an isomorphism invariant group problem, then $L$ is not $NP$-complete under ETH.
\end{theorem}

\medskip

\begin{remark}
Due to Buhrman and Homer \cite{10.1007/3-540-56287-7_99}, under $EXP\neq NEXP$, no quasipolynomially sparse language is NP-complete. Hence, the isomorphism invariant problems on both group-based graphs (\Cref{iso-inv}) and groups (\Cref{iso-inv-group}) are not NP-complete, under the assumption $EXP\neq NEXP$.  
\end{remark}


\section{Hardness of the \textsc{Weighted Max-Cut} Problem and \textsc{Graph Motif} Problem}\label{sec: motif}

First, we study the \textsc{Weighted Max-Cut} problem: given an edge-weighted graph $\Gamma$ and a number $k$, the goal is to find a subset $S \subset V(\Gamma)$ such that the total weight of edges across $S$ and $V(\Gamma)\setminus S$ is at least $k$.  


\medskip 

\begin{theorem}\label{max-cut np}
    \textsc{Weighted Max-Cut} is NP-complete for power graphs, even when the underlying groups are cyclic.
\end{theorem}
\begin{proof}
     \textsc{Weighted Max-cut} can be easily reduced to  \textsc{Weighted Max-Cut} for complete graphs. 
     We reduce \textsc{Weighted Max-Cut} for complete graphs to \textsc{Weighted Max-Cut} in power graphs of cyclic groups. From an instance of \textsc{Weighted Max-Cut} in $K_n$ (complete graph of size $n$), we construct an instance of \textsc{Weighted Max-Cut} in $Pow(\mathbb{Z}_{n^2})$. From \Cref{equivalence classes form a clique}, the generators of $\mathbb{Z}_{n^2}$ forms a clique of size $\phi(n^2)$. Since $\phi(n^2)= n \cdot \phi(n) 
     > n$ (as $n > 1$), the graph $Pow(\mathbb{Z}_{n^2})$ contains an induced subgraph $\Gamma'$ that is isomorphic to $K_n$. Hence, any edge-weighted $K_n$ can be embedded in edge-weighted $Pow(\mathbb{Z}_{n^2})$ by setting weights of all the edges which are not in $\Gamma'$ as $0$.
      \end{proof}

We now shift our attention to the \textsc{Graph Motif} problem. 
A \textit{motif} $M$ is a multiset of colours; it is \textit{colourful} if it is a set. In a vertex-coloured graph $\Gamma$, the colour multiset $col(V')$ of a set $V' \subseteq V(\G)$ is the multiset of the colours of the vertices in $V'$.
We say that $M$ \emph{occurs} in $\Gamma$ if there exists a connected subgraph induced on $V' \subseteq V(\Gamma)$ such that $col(V') = M$, in which case $V'$ is an \emph{occurrence} of $M$. The \textsc{Graph Motif} problem asks, given a vertex-coloured graph $\Gamma$ and a motif $M$, whether such an occurrence $V'$ exists.


 \medskip 

\begin{theorem}\label{Graph motif not in P}
    Under ETH, \textsc{Graph Motif} problem on power graphs cannot be solved in time $2^{O(\log ^c n)}$, for any constant  $c\geq 0$, even if the underlying group is cyclic and $M$ is colourful.
\end{theorem}

Using the next lemma, \Cref{Graph motif not in P} can be easily proved.

\medskip 

\begin{lemma}\label{sat to pow}
    There is an $n^{O( \log (\log n))}$-time many-one reduction from $3$-SAT to \textsc{Graph Motif} in power graphs of cyclic groups.
\end{lemma}
\begin{proof} The proof is divided into two main steps:\\
 (1) Given a $3$-SAT formula ${\Phi}$, with $n$ variables, we construct an instance $(\GG,M)$ of \textsc{Graph Motif} in a vertex-coloured graph $\GG$ in polynomial time such that $\Phi$ is satisfiable if and only if $M$ occurs in $\GG$.\\
 (2) We embed $\GG$ as an \emph{induced subgraph} in $Pow(\ZZ_N)$, where $N=n^{O(\log(\log n))}$.


 
\noindent\textit{Step 1:} Let $\Phi=C_1 \wedge C_2 \wedge \dots \wedge C_m$ be an instance of $3$-SAT over variables $x_1,x_2,\dots,x_n$. 
To construct $\GG$, we first create a vertex $r$ and colour it with $0$. Then for each literal $l_i \in \{x_i,\overline{x}_i\}$, $1 \leq i \leq n$, we create two vertices $x_i$ and $\overline{x}_i$ both coloured with $i$. After that, we add the following edges: $\{r,x_i\}$, $\{r,\overline{x}_i\}$, for $1 \leq i \leq n$. Now for each literal $l_i$, we do the following: Suppose $l_i$ appears in clauses $C_{i_1}, C_{i_2},\dots, C_{i_s}$; then create a clique consisting of vertices $z_{i_1}, z_{i_2}, \dots, z_{i_s}$ that are coloured with $C_{i_1}, C_{i_2},\dots, C_{i_s}$ respectively, and add the edges: $\{l_i,z_{i_j}\}$, $1\leq j \leq s$. Note that a new set of $z$-vertices is added for vertices corresponding to each literal. Therefore, the total number of added $z_{i_j}$ vertices is $3m$. The motif $M$ is specified to be the set of colours $\{0,1,2,\dots,n,C_1,C_2,\dots,C_m\}$. 

\begin{figure}[hpt!]
    \centering
    \begin{tikzpicture}[scale=0.3]
        
        \coordinate (O) at (12,11);
        
        \coordinate (x1) at (2,6);
        \coordinate (x1n) at (6,6);
        \coordinate (x2) at (10,6);
        \coordinate (x2n) at (14,6);
        \coordinate (x3) at (18,6);
        \coordinate (x3n) at (22,6);
        
        \coordinate (y1) at (1,1);
        \coordinate (y2) at (3,1);
        \coordinate (y3) at (5,1);
        \coordinate (y4) at (7,1);
        \coordinate (y5) at (9,1);
        \coordinate (y6) at (11,1);
        \coordinate (y7) at (10,-1);
        \coordinate (y8) at (14,1);
        \coordinate (y9) at (17,1);
        \coordinate (y10) at (19,1);
        \coordinate (y11) at (18,-1);
        \coordinate (y12) at (22,1);
        
       \draw (O) -- (x1);
       \draw (O) -- (x1n);
       \draw (O) -- (x2);
       \draw (O) -- (x2n);
       \draw (O) -- (x3);
       \draw (O) -- (x3n);
       
       \draw (x1) -- (y1);
       \draw (x1) -- (y2);
       \draw (x1n) -- (y3);
       \draw (x1n) -- (y4);
       \draw (x2) -- (y5);
       \draw (x2) -- (y6);
       \draw (x2n) -- (y8);
       \draw (x3) -- (y9);
       \draw (x3) -- (y10);
       \draw (x3n) -- (y12);
       
       \draw (y1) -- (y2);
       \draw (y3) -- (y4);
       \draw (y5) -- (y6);
       \draw (y9) -- (y10);
       \draw (y11) -- (y9);
       \draw (y11) -- (x3);
       \draw (y11) -- (y10);
       \draw (y5) -- (y7);
       \draw (x2) -- (y7);
       \draw (y6) -- (y7);

       \draw (O) -- (x2);
       \draw (O) -- (x2n);
       \draw (O) -- (x3);
       \draw (O) -- (x3n);

        \filldraw [brown] (O) circle(10pt);
        \filldraw [magenta] (x1) circle(8pt);
        \filldraw [magenta] (x1n) circle(8pt);
        \filldraw [magenta] (x2) circle(8pt);
        \filldraw [magenta] (x2n) circle(8pt);
        \filldraw [magenta] (x3) circle(8pt);
        \filldraw [magenta] (x3n) circle(8pt);
        
        \filldraw [cyan] (y1) circle(8pt);
        \filldraw [cyan] (y2) circle(8pt);
        \filldraw [cyan] (y3) circle(8pt);
        \filldraw [cyan] (y4) circle(8pt);
        \filldraw [cyan] (y5) circle(8pt);
        \filldraw [cyan] (y6) circle(8pt);
        \filldraw [cyan] (y7) circle(8pt);
        \filldraw [cyan] (y8) circle(8pt);
        \filldraw [cyan] (y9) circle(8pt);
        \filldraw [cyan] (y10) circle(8pt);
        \filldraw [cyan] (y11) circle(8pt);
        \filldraw [cyan] (y12) circle(8pt);

        \draw (O) node[above, yshift=3pt] {$0$};
        \draw (x1) node[left] {$1$};
        \draw (x1n) node[left] {$1$};
        \draw (x2) node[left, xshift=-2 pt] {$2$};
        \draw (x2n) node[right,xshift=2pt] {$2$}; 
        \draw (x3) node[right,xshift=2pt] {$3$};
        \draw (x3n) node[right,xshift=2pt] {$3$}; 
        
        \draw (y1) node[below, yshift=-2pt] {$C_1$};
        \draw (y2) node[below, yshift=-2pt] {$C_3$};
        \draw (y3) node[below, yshift=-2pt] {$C_4$};
        \draw (y4) node[below, yshift=-2pt] {$C_2$};
        \draw (y5) node[below, yshift=-2pt] {$C_1$}; 
        \draw (y6) node[below, xshift=2 pt, yshift=-2pt] {$C_2$};
        \draw (y7) node[below, yshift=-2pt] {$C_4$};
        \draw (y8) node[below, yshift=-2pt] {$C_3$};
        \draw (y9) node[below, yshift=-2pt] {$C_1$};
        \draw (y10) node[below, xshift=3 pt, yshift=-2pt] {$C_3$};
        \draw (y11) node[below, yshift=-2pt] {$C_4$}; 
        \draw (y12) node[below, yshift=-2pt] {$C_2$};
              
         \end{tikzpicture}
        \caption{\footnotesize{$C_1=(x_1\lor x_2 \lor x_3)$, $C_2=(\overline{x}_1\lor x_2 \lor \overline{x}_3)$, $C_3=(x_1\lor \overline{x}_2 \lor x_3)$, $C_4=(\overline{x}_1\lor x_2 \lor x_3)$}}
\end{figure}


    Since the number of clauses in ${\Phi}$ is $O(n^3)$ (as ${\Phi}$ is a $3$-SAT), the construction of $\GG$ above is clearly polynomial in $n$. The correctness of this construction is stated in the next claim (See  \Cref{app: proof of claim 1(sat to star-clique)} for the proof of the claim).

\begin{claimc}
    \label{sat to star-clique}
        $\Phi$ is satisfiable if and only if $M$ occurs in $\GG$.
\end{claimc}

\noindent \textit{Step 2:} We embed $\GG$ as an induced subgraph in $Pow(\mathbb{Z}_N)$ for some suitably chosen $N$. The number $N$ is the product of the first $b$ primes, where $b$ is an odd number. We later show that setting $b=O(\log n)$ is sufficient. It is a well-known fact that the product of the first $b$ primes is bounded above by $2^{(1+o(1))b \log b}$ (see, e.g., \cite{finch2003mathematical}). With that, we have $N=2^{O(\log n \log (\log n))}$.

First, we describe the embedding procedure. 
Let $N=p_1p_2\dots p_b$, where $p_i$ is the $i$-th natural prime number. To embed $\GG$ in $Pow(\mathbb{Z}_N$), we define a bijective mapping $f : V(\Gamma_{\phi}) \xrightarrow{} S$, where $S \subseteq \mathbb{Z}_N$, such that $f$ is an isomorphism between $\GG$ and $Pow(\mathbb{Z}_N)[S]$. 

Let $\PP_{\frac{b-1}{2}}(\{2,3,\dots,b\})$ denote the collection of subsets $I \subset \{2,3,\dots,b\}$ such that $|I|=\frac{b-1}{2}$. Let $f_{\PP}$ be a one-to-one function 
from the set of literals of $\Phi$ to $\PP_{\frac{b-1}{2}}(\{2,3,\dots,b\})$. The function $f$ maps the vertex $x_i$ (or $\overline{x}_i$) to an element $g \in \mathbb{Z}_N$ of order \(p_1 \prod_{t\in f_{\PP}(l_i)}p_t\), where $l_i=x_i$ (or $l_i=\overline{x}_i$). The element $g$ has a neighbour $h$ of order $\prod_{t\in f_{\PP}(l_i)}p_t$ (by \Cref{remark power graph}). Moreover, $g$ is adjacent to every element in the equivalence class $[h]$ (by \Cref{equivalence classes form a clique}). Suppose $N(u)=\{r,z_{i_1},z_{i_2},\dots,z_{i_k}\}$, where $u=x_i$ or $\overline{x}_i$. Then, the function $f$ maps each $z_{i_j}$, $1\leq j \leq k$, to distinct elements in $[h]$. For this mapping procedure to be one-to-one, we need $|[h]| \geq |N(u)|-1$. The function $f$ maps the vertex $r$ to some element of order $p_1$.

Once we ensure that $f$ is one-to-one, then taking $S=f(V(\GG))$ implies that $f$ is a bijection from $V(\GG)$ to $S$. The next claim shows that taking $b=O(\log n)$ is sufficient for the mapping $f$ to be one-to-one, and consequently to be an isomorphism between $\GG$ and $Pow(\mathbb{Z}_N)[S]$.

\begin{claimc}\label{star-clique to pow}
    If $b=O(\log n)$, then $f$ is an isomorphism between $\GG$ and $Pow(\mathbb{Z}_N)[S]$.
\end{claimc}




\begin{proof}

We begin by identifying two inequalities, that must be satisfied for \( f \) to be one-to-one. Subsequently, we show that \( b = O(\log n) \) satisfies these inequalities.

(1) First, $f_\mathcal{P}$ has to be one-to-one. For that, we need the following: 
\begin{equation}\label{variables_1}
\binom{b-1}{\frac{b-1}{2}} \geq 2n
\end{equation}

(2)  For each vertex $u \in \{x_i,\overline{x}_i \ | \ 1 \leq i \leq n \}$, we must have: $|[h]| \geq |N(u)| - 1$, where $h$ is a neighbour of $f(u)$ in $Pow(\mathbb{Z}_N)$. Note that $|N(u)|$ is at most the number of clauses $m$ in $\Phi$ and $m=O(n^3)$. From the definition of $f$, one can see that $|[h]|$ is at least $\phi(p_2p_3\dots p_{\frac{b+1}{2}})$, as $|[h]|=\phi(ord(h))$. Hence, 
satisfying the following inequality is sufficient: 

\begin{equation}\label{clauses2}
    cn^3 \leq \phi(p_2p_3\dots p_{\frac{b+1}{2}}), \text{ where $c$ is some constant}    
\end{equation}

We now show that for $b=2\log(2n)+1$ and for sufficiently large $n$, both the above inequalities, \Cref{variables_1} and \Cref{clauses2}, 
are satisfied.


Note: $\phi(p_2p_3\dots p_{\frac{b+1}{2}}) = 2.4.6.\dots.(p_{\frac{b+1}{2}}-1)$ $\geq$ $ 2.3.5. \dots. p_{\frac{b-1}{2}}$ $\geq 2^{c'{\frac{b-1}{2}}. \log({\frac{b-1}{2}})}$, for some fixed constant $c'$. The last inequality can be easily derived from the following two facts: (1) $lim_{n\to\infty}(p_1p_2\dots p_n)^{\frac{1}{p_n}}=e$ (see \cite{finch2003mathematical}); (2) prime number theorem, i.e., $\lim_{n\to\infty}\frac{p_n}{n \log n}=1$ (see, e.g., \cite{apostol2013introduction}). Hence, if 
the following inequality holds, then \Cref{clauses2} consequently holds:

\begin{equation}\label{clauses1_1}
    cn^3 \leq 2^{c'{\frac{b-1}{2}}. \log({\frac{b-1}{2}})}
\end{equation}

Now, we find a suitable value of $b$ such that \Cref{variables_1} and \Cref{clauses1_1} simultaneously hold. Let us denote $(b-1)/2$ by $k$. Using Stirling's approximation, 
for sufficiently large $k$, $\binom{2k}{k} \sim \frac{2^{2k}}{\sqrt{k \pi}}$. Then \Cref{variables_1} implies the following inequality, and solving it serves our purpose.
\begin{equation}\label{stirling}
\begin{split}
 \frac{2^{2k}}{\sqrt{k \pi}} & \geq 2n \\
\end{split}
\end{equation}
For $k=\log (2n)$, \Cref{stirling} holds, due to the inequality: $2 \log (2n) - \frac{1}{2} \log (\log (2n)\sqrt{\pi}) 
    \geq \log(2n)$. 
Moreover, for $b=2\log(2n)+1$ and for sufficiently large $n$, \Cref{clauses1_1} is also satisfied.



Therefore, $f$ is a bijection between $V(\GG)$ and $S\subseteq \mathbb{Z}_N$, where $S=f(V(\GG))$ and $N$ is the product of first $b$ primes. We now show that $f$ is an isomorphism between $\GG$ and $Pow(\mathbb{Z}_N)[S]$. 

Let $\{u,w\}$ be an edge in $\GG$. We show that $\{f(u),f(w)\}$ is an edge in $Pow(\mathbb{Z}_N)$ considering three nonequivalent cases. The first case is when $col(u)=C_i$ and $col(w)=C_j$ for some $1 \leq i < j \leq m$. Since $\{u,w\}$ is an edge, $u$ and $w$ must have a common neighbour coloured with $k$ 
for some $1 \leq k \leq n$. Therefore, by the definition of $f$, $f(u)$ and $f(w)$ both belong to the same equivalence class and hence are adjacent in $Pow(\mathbb{Z}_N)$ (by \Cref{equivalence classes form a clique}). The second case is when $col(u)=i$ 
and $col(w)=C_j$ for some $1\leq i \leq n$ and $1 \leq j \leq m$. By the definition of $f$, $ord(f(u))=p_1\prod_{t\in f_{\PP}(l)}p_t$, where $l$ is the literal corresponding to $u$ and $ord(f(w))=\prod_{t\in f_{\PP}(l)}p_t$; hence, they are adjacent in $Pow(\mathbb{Z}_N)$ (by \Cref{remark power graph}). The only remaining case is when $col(u)=0$ 
and $col(w)=i$. 
Similarly to the second case, it can be shown that $\{f(u),f(w)\}$ is an edge in $Pow(\mathbb{Z}_N)$. 
For the other direction, we show that if $\{u,w\}$ is not an edge in $\GG$, then $\{f(u),f(w)\}$ is not an edge in $Pow(\mathbb{Z_N})$. The case when $col(u)=0$ and $col(w)=C_j$ for some $1 \leq j \leq m$ is easy to show: as in this case, neither $ord(f(u))$ divides $ord(f(w))$ nor $ord(f(w))$  divides $ord(f(u))$ and hence are not adjacent in $Pow(\mathbb{Z}_N)$ (by \Cref{remark power graph}). For the second case, without loss of generality, we assume that $u=x_i$ and $w=z_{j_k}$, such that the literal $x_i$ does not appear in clause $C_{j_k}$ (where $col(z_{j_k})=C_{j_k}$). Then, $ord(f(u))=p_1\prod_{t \in f_{\PP}(x_i)}p_t$ and $ord(f(w))=\prod_{t \in f_{\PP}(l_j)}p_t$, where $l_j$ is some literal appearing in the clause $C_{j_k}$. Since $l_j \neq x_i$ and $f_{\PP}$ is one-to-one, $f_{\PP}(x_i)$ and $f_{\PP}(l_j)$ are distinct elements of $\PP_{\frac{b-1}{2}}(\{2,3,\dots,b\})$. Hence, neither $ord(f(u))$ divides $ord(f(w))$ nor $ord(f(w))$ divides $ord(f(u))$, and therefore there is no edge between $f(u)$ and $f(w)$ in $Pow(\mathbb{Z}_N)$ (by \Cref{remark power graph}). The remaining non-equivalent cases are as follows: one is when $u=z_{j_k}$ and $w=z_{j_{k'}'}$ such that no literal $x_i$ or $x'_i$ appears in both $C_{j_k}$ and $C_{j'_{k'}}$, where $C_{j_k}=col(z_{j_k})$ and $C_{j'_{k'}}=col(z_{j'_{k'}})$;
 another one is when $u=x_i$ and $w=x_j$, where $i \neq j$. Note that, similar to the second case, in these two cases also, one can show that $\{f(u),f(w)\}$ is not an edge in $Pow(\mathbb{Z}_N)$.
 
 Hence, the claim is proved. 
\end{proof}

We note that there exists a one-to-one mapping $f_{\PP}$, as required, computable in time $\binom{b-1}{(b-1)/2}$, implying that $f$ can be computed in $2^{O(\log n \log (\log n))}$ time. Following the embedding of $\GG$ in \Cref{star-clique to pow}, an instance of \textsc{Graph Motif} in $\text{Pow}(\mathbb{Z}_N)$ can be constructed by assigning a colour not in $M$ to all vertices in $\mathbb{Z}_N \setminus S$. This yields an $n^{O(\log \log n)}$-time reduction from 3-SAT to \textsc{Graph Motif} in power graphs of cyclic groups, completing the proof of the lemma.
\end{proof}



We are now ready to prove \Cref{Graph motif not in P}.

\begin{proofof}{\Cref{Graph motif not in P}}
        The proof is by contradiction. 
    Assume that \textsc{Graph Motif} in the power graph of the cyclic group of order $N$ can be solved in $2^{O((\log N)^c)}$ time. Given an instance $\Phi$ $3$-SAT with $n$ variables, using the reduction stated in \Cref{sat to pow}, we get an instance of \textsc{Graph Motif} in the power graph $Pow(\mathbb{Z}_N)$ such that $N=2^{O(\log n \log(\log n))}$. Since $2^{O((\log n)^c (\log(\log n)^c)}=2^{o(n)}$, there is an algorithm of $\Phi$ with runtime $2^{o(n)}$, which contradicts ETH. Hence, our assumption is wrong. \hfill$\blacktriangleleft$
\end{proofof}

\medskip 

\begin{remark}
    A similar statement to \Cref{Graph motif not in P} can be made without using ETH. More explicitly, \textsc{Graph Motif} on power graphs cannot be solved in quasipolynomial time, unless $\mathsf{EXP}=\mathsf{NEXP}$. This is because, in that case \Cref{sat to pow} would imply that $\mathsf{NP} \subseteq \bigcup_{t>0}\mathsf{DTIME}(2^{\log ^tn)})$, which in turn would imply that $\mathsf{EXP}=\mathsf{NEXP}$ \cite{10.1007/3-540-56287-7_99}.
\end{remark}

We now show the existence of an algorithm with runtime $2^{o(n)}$ for power graphs of cyclic groups. We begin with a few definitions and a relevant result.
    Two vertices $v$ and $v'$ of a graph are said to have the same \emph{type} if and only if $N(v)\setminus\{v'\} = N(v')\setminus \{v\}$. A graph $\Gamma(V,E)$ has \emph{neighbourhood diversity} at most $w$ if there exists a partition of $V(\G)$ into at most $w$ sets, such that all the vertices in each set have the same type \cite{lampis2012algorithmic}.
Ganian \cite{ganian2012using} gave an FPT algorithm for \textsc{Graph Motif} parameterised by neighbourhood diversity.

\medskip 

\begin{lemma}\cite{ganian2012using}\label{ganian}
     \textsc{Graph Motif} can be solved in time $O(2^k\cdot \sqrt{|V|}|E|)$ on a graph $\Gamma=(V,E)$ of neighbourhood diversity at most $k$.
\end{lemma}

Let $g$ and $g'$ be two elements in a group $G$. If $[g]=[g']$, then  $g$ and $g'$ have the same type in $Pow(G)$. Hence, the neighbourhood diversity of $Pow(G)$ is at most the number of equivalence classes of $G$. If $G$ is a cyclic group of order $n$, then the number of equivalence classes is the number of total divisors of $n$, which is bounded by $n^{O(\frac{1}{\log\log n})}$ \cite{apostol2013introduction}. Hence, we have the following corollary.

\medskip 

\begin{corollary}\label{cyclic group graph motif}
     \textsc{Graph Motif} for $Pow(\ZZ_n)$ can be solved in $2^{o(n)}$ time, in particular, in 
     $2^{n^{c/\log \log n}}$ time.
\end{corollary}

From \Cref{Graph motif not in P}, we know that \textsc{Graph Motif} cannot be solved in quasipolynomial time for the power graphs of finite cyclic groups. On the other hand, \Cref{cyclic group graph motif} shows the existence of a $2^{o(n)}$ time algorithm for \textsc{Graph Motif} of power graphs of cyclic groups. \Cref{cyclic group graph motif} also implies that there cannot be a linear time many-one reduction from $3$-SAT to \textsc{Graph Motif} on power graphs of cyclic groups under ETH. 

When $G=(\ZZ_p)^n$, where $p$ is a prime, the number of equivalence classes 
 is $(p^n-1)/(p-1)$, and hence Ganian's algorithm (see \Cref{ganian}) yields running time $2^{O(|G|/p)}$, which is almost exponential. Hence, the parameterised algorithm is not efficient in the case of power graphs of $p$-groups. We design a polynomial time algorithm for solving \textsc{Graph Motif} in power graphs of $p$-groups. 


\medskip 

\begin{theorem} \label{non-cyclic poly motif}
    \textit{\textsc{Graph Motif}} for power graphs of $p$-groups can be solved in polynomial time. 
\end{theorem}
\begin{proof}

  Since the power graph of a cyclic $p$-group is a complete graph (see \Cref{complete power}), \textsc{Graph Motif} can be solved in linear time in this case. 
  Let $\Gamma=Pow(G)$, where $G$ is a non-cyclic $p$-group. We design a greedy algorithm.
  
  At first, we delete the set $S$ of vertices which have colours outside the motif $M$. Let $\Gamma'$ denote the graph $\G \setminus S$, i.e., the graph obtained by deleting all the vertices of $S$ from $\G$. As $\Gamma'$ might be disconnected, we first find the connected components of $\Gamma'$, and 
  check whether any of these connected components contain an occurrence of $M$.
  Let $C$ be the connected component being inspected at some stage. We select any dominating vertex\footnote{A vertex $u$ is said to be a \textit{dominating vertex} of a graph $\G$, if $N(u) \cup \{u\}=V(\G)$. } $v$  of $C$ (the existence of such a dominating vertex is proved in \Cref{lem: pgroup dominating} 
  ). If $M$ is a subset of $\{col(v)\}\cup col(N_C(v))$, where $N_C(v)$ denotes the neighbourhood of $v$ in $\Gamma'[C]$,
  we return `Yes'
  ; else, we check the next unchecked connected component. If none of the connected components contains an occurrence of $M$, we return `No'. 
  It is easy to see that the runtime of our algorithm is $poly(|G|)$. 

Before proving the correctness of the algorithm, we first prove the next claim.

\begin{claimd} \label{lem: pgroup dominating}
 There is at least one dominating vertex in any connected component $C$ of $\Gamma'$.
\end{claimd}
\begin{proof}

   In a power graph, if there exists an edge $\{g,h\} \in E$ such that $ord(g)=p^i$, $ord(h)=p^j$ and $i \leq j$, then $g \in \langle h \rangle$. 
   Let $x$ be a vertex in $C$ such that $ord(x)=p^l$, where $l$ is the lowest power of $p$ in $C$. We claim that $x$ is a dominating vertex of $\Gamma'[C]$, which we prove by contradiction.

   Assume that there exists a vertex $y \in C$ such that $ord(y)=p^j$, $j \geq l$ and $x$ is not adjacent to $y$. As $\Gamma'[C]$ is connected, there exists a path between $x$ and $y$ of length at least $2$. Let $xx_1x_2\dots x_ky$ be a shortest path between $x$ and $y$. 
   This implies $\{x,x_1\} \in E(\Gamma)$ and $\{x,x_2\} \notin E(\Gamma)$. 
   Let $ord(x_1)=p^{l'}$, where $l' \geq l$. So, $\{x,x_1\} \in E(\Gamma)$ implies $x \in \langle x_1 \rangle$. Now, $\{x_1,x_2\} \in E(\Gamma)$ implies that either $x_1 \in \langle x_2 \rangle$ or $x_2 \in \langle x_1 \rangle$. If $x_1 \in \langle x_2 \rangle$, then $x \in \langle x_2 \rangle$, which is a contradiction. So, let $x_2 \in \langle x_1 \rangle$. Notice that $\langle x_1 \rangle \cap C$ is a clique in $\Gamma'$ (by \Cref{complete power}). This implies $x,x_2 \in \langle x_1 \rangle$ and hence $\{x,x_2\} \in E(\Gamma)$, which is a contradiction. 
   Therefore, our assumption 
   is wrong. Hence, $x$ is a dominating vertex of $C$.
\end{proof}

For the correctness of our algorithm, we need to argue that, if there is an occurrence of $M$ in $\Gamma$, then there is a connected component $C$ in $\Gamma'$ such that the occurrence of $M$ in $\Gamma'[C]$ contains a dominating vertex of $C$. Note that the colour of the dominating vertex $v$ is in $M$; otherwise, it would have been deleted. 
Let $W$ be any occurrence of $M$ in $\Gamma$. Then, $\Gamma[W]=\Gamma'[W]$ and hence, $W$ is a subset of some connected component $C$ in $\Gamma'$. Let $v$ be a dominating vertex of $C$. If $W$ contains $v$, then we have nothing to prove. Suppose $W$ does not contain $v$. As $v\in \Gamma'$, $col(v) \in M$. Moreover, there must be some vertex $w \in W$ such that $col(v)=col(w)$. So, $W'=(W \setminus \{w\}) \cup \{v\}$ is an occurrence of $M$ in $\Gamma$ such that $W'$ contains a dominating vertex $v$ of the connected component $C$ of $\Gamma'$.
\end{proof}

\section{Recognition of Power Graphs}\label{sec: recognition}

In this section, we show that the recognition problem for power graphs of abelian groups and certain nilpotent groups can be solved in polynomial time (\Cref{summary}). The input graph is given in the standard graph representation model. The task is to decide if the graph is a power graph of some group. 

Das et al. \cite{das_et_al:LIPIcs.FSTTCS.2024.20} gave a reconstruction algorithm that takes a power graph as input and produces a directed power graph. However, if the input is not a power graph, the algorithm may fail or produce a directed graph that is not a directed power graph. So, we may, \emph{without loss of generality, focus on the recognition problem of directed power graphs}.

The quasipolynomial time recognition algorithm by Arvind, Cameron, Ma, and Maslova \cite{arvind2024aspectscommutinggraph} uses the concept of \emph{description} of groups. The \textit{description} of a finite group \( G \), denoted as $desc(G)$, is the encoding of $G$ into a binary string\footnote{Note that description of a group is different from the well-known notion of presentation of a group, which is defined and used later in this section.} of length at most \( c \cdot \log^3|G| \), where $c>0$ is some fixed constant. They designed an algorithm $\mathcal{D}$ that, on input $desc(G)$, computes the multiplication table of a group $G'$ isomorphic to $G$ in time \( poly(|G|)\). Moreover, using the algorithm $\mathcal{D}$, the multiplication tables of all groups G of order n are generated as follows: First, enumerate all possible binary strings $w$ of length at most $c\cdot \log^3 n$. Then, for each $w$, run $\mathcal{D}$ on it. If $w = desc(G)$ for some $G$, the algorithm $\mathcal{D}$ outputs its multiplication table. From that multiplication table, the commuting graph of $G$ can be found in time polynomial in $|G|$. Then, using Babai's graph isomorphism algorithm \cite{babai2016graph}, it can be checked whether a given graph $\Gamma$ is isomorphic to the commuting graph of $G$ in time $2 ^{O(\log^c n)}$. The same approach works for the recognition of any graph defined in terms of groups.


\medskip 

\begin{restatable}{lemma}{fta}\label{subexpo}
 There are algorithms that, on the input of an undirected graph $\G$ (or directed graph $\DG$) of size n, check whether $\G$ (or $\DG$) is isomorphic to a power graph (or directed power graph), enhanced power graph, deep commuting graph or generating graph in $2^{O(\log ^c n)}$ time. 
\end{restatable}


 Two aspects of the algorithm by Arvind et al. that contribute to high runtime are: (a) the number of non-isomorphic groups of order $n$ can be $2^{\Omega(\log^3n)}$ \cite{mciver1987enumerating}, (b) Babai's isomorphism algorithm has quasipolynomial runtime. The polynomial time algorithm for the isomorphism of power graphs for nilpotent groups by Das et al. \cite{das_et_al:LIPIcs.FSTTCS.2024.20} appears to be useful; 
but, it solves a promise problem: it is guaranteed to work correctly if the graphs are power graphs of nilpotent groups. However, if one input is not a power graph of a nilpotent group, it may give a false positive answer (for details, see \Cref{story behind the two techniques}). 
To tackle this issue, we use the following two tools: (i) the notion of \emph{reduced graphs} (similar to \cite{das_et_al:LIPIcs.FSTTCS.2024.20}); (ii) an algorithm referred to as \emph{Gluing process}. Using these, we prove that the recognition problem of directed power graphs of nilpotent groups is reducible to the recognition problem of directed power graphs of $p$-groups (\Cref{pow(nil) reducible to dpow(p)}).

\noindent\textbf{Reduced graph:} 
Given a directed graph $\DG$, we construct a directed graph $R(\DG)$ from $\DG$, referred to as the \emph{reduced graph}, by the following method: (1) For each closed twin-class $\tau$ of $R(\DG)$,  we create a vertex in $\DG$; (2) The colour of each vertex $\tau$, denoted by $col(\tau)$, of $R(\DG)$ is the out-degree of any vertex of $\tau$ in $\DG$; (3) Put a directed edge from $\tau_1$ to $\tau_2$ in $R(\DG)$ where $\tau_1 \neq \tau_2$, if there is a directed edge from some vertex $x_1 \in \tau_1$ to some vertex $x_2 \in \tau_2$ in $\DG$. 

 Recall that for any finite group $G$, the out-degree of a vertex $u$ in $DPow(G)$ equals to $ord(u)$ in $G$. If $(u,v)$ is an edge in $DPow(G)$, then $ord(v) \mid ord(u)$. Hence, if $(\tau_1,\tau_2)$ is an edge in $R(DPow(G))$, then $col(\tau_2) \mid col(\tau_1)$. So, it is easy to see that $R(DPow(G))$ is a directed acyclic graph (DAG). In $DPow(G)$, two vertices $u$ and $v$ are closed twins if and only if $\langle u \rangle =\langle v \rangle $ in $G$, i.e., $u$ and $v$ are two generators of the same cyclic subgroup in $G$. So, the vertex set of  $R(DPow(G))$ is $\{[x] \ | \ x \in G \}$. If $(x,y)$ is an edge in $DPow(G)$, then for every $x'\in [x]$ and for every $y' \in [y]$, there is an edge $(x',y') \in DPow(G)$. Then by \Cref{remark power graph}, there is an edge $(\tau_1,\tau_2)$ in $R(DPow(G))$. 
 From these properties, the following lemma from \cite{das_et_al:LIPIcs.FSTTCS.2024.20} is not hard to see.

\medskip 

\begin{lemma}\label{color iso in enough}
 Let $DPow(G)$ and $DPow(H)$ be two directed power graphs. Then $R(DPow(G))$ is colour-isomorphic to $R(DPow(H))$ if and only if $DPow(G) \cong DPow(H)$. 
\end{lemma}


Due to the above lemma, for the rest of the discussion, \emph{we assume that the input to our recognition query is a reduced directed graph}. First, we consider the subproblem with inputs, a reduced directed graph $\DDG{1}$ and a nilpotent group $H$; the task is to check whether $\DDG{1} \cong_c R(DPow(H))$. When the input group is a $p$-group, then we can directly use the linear time isomorphism algorithm of directed power graphs of $p$-groups by \cite{das_et_al:LIPIcs.FSTTCS.2024.20}. See the next lemma for the statement.


\medskip 

\begin{lemma}\label{iso: nil}   Let $\DDG{1}$ be a reduced directed graph and $H$ be a finite $p$-group. Then, there is a linear time algorithm that checks whether $\DDG{1} \cong_c R(DPow(H))$ and if so, returns a colour-preserving isomorphism $\psi:V(R(DPow(H)))\xrightarrow[]{} V(\DDG{1})$.
\end{lemma}

    In a vertex-coloured directed acyclic graph, let vertices $u_1,u_2,\dots,u_k$ have a common parent $u$. Then $u$ is called \emph{least common parent} of $u_1,u_2,\dots,u_k$, if for any other existing common parent $u'$ of $u_1,u_2,\dots,u_k$, $col(u) \leq col(u')$.

Now consider the case when the input group is a finite nilpotent group $H=H_1\times H_2 \times \dots \times H_k $, where $H_i$ is the Sylow $p_i$-subgroup of $H$. Note that any group element of $H$ can be represented as $(x_1,x_2,\dots,x_k)$, where $x_i\in H_i$, $1\leq i \leq k$. We now describe the \emph{Gluing Process}.

\noindent\textbf{Gluing Process:} Let $\DDG{1}$ be a reduced directed graph and $H=H_1 \times \dots \times H_k $ be a finite nilpotent group. 
Suppose $\DDG{1i}$ denotes the subgraph of $\DDG{1}$ induced on the set of vertices whose colours are $p_i$-powers, for all $1\leq i \leq k$. 
Let $f_i : V(R(DPow(H_i)))  \xrightarrow{} V(\DDG{1i})$ be a colour-isomorphim  between $R(DPow(H_i))$ and $\DDG{1i}$, for $1\leq i \leq k$. Gluing process decides if there exists a colour-isomorphism $f: V(R(DPow(H))) \xrightarrow{} V(\DDG{1})$ such that $f$ restricted to $V(R(DPow(H_i)))$ is $f_i$. 
To achieve this, we define $f: V(R(DPow(H))) \xrightarrow{} V(\DDG{1})$ 
as follows: $f$ maps each vertex $[(x_1,\dots,x_k)]$ of $R(DPow(H))$ to a least common parent, if it exists, of $f_1([x_1])$, $\dots$, $f_k([x_k])$ in $\DDG{1}$. 
Thereafter, we check if $f$ is a colour-isomorphism between $R(DPow(H))$ and $\DDG{1}$. For a detailed discussion of the correctness of Gluing process, we refer the reader to \Cref{correctness of patchup}, particularly \Cref{rong mile geche}. 
Now, we prove one of the main results of this section. 

\medskip

\begin{restatable}{lemma}{ftaa}
\label{pow(nil) reducible to dpow(p)}
   Let $\mathcal{X}$ be a subclass of $\mathcal{G}_{PP}=\{G \text{ is a $p$-group}: \text{ $p$ is some prime}\}$ and $Nil(\mathcal{X})$ be the class of nilpotent groups whose Sylow subgroups are in $\mathcal{X}$. Let $\mathcal{A}$ be a polynomial time algorithm that on input a reduced directed graph $\mathcal{D}'$, decides if $\mathcal{D}' \cong_c R(DPow(G))$, for some $G\in \mathcal{X}$ and if so, returns a colour-preserving isomorphism $\psi:V(R(DPow(G))) \xrightarrow[]{} V(\mathcal{D}')$ along with $G$. Then, there is a polynomial time algorithm $\mathcal{A}_{nil}$, that, on input a reduced directed graph with $n$ vertices, checks whether it is colour-isomorphic to $R(DPow(H))$ for some group $H \in Nil(\mathcal{X})$. Moreover, $\mathcal{A}$ is called  $O(n)$ times by $\mathcal{A}_{nil}$. 
\end{restatable}
\begin{proof}
    Let $\DG$ be a reduced directed graph with $n$ vertices, where \newline $n=p_1^{\alpha_1}\dots p_k^{\alpha_k}$, $\alpha_i \geq 1$. The algorithm $\mathcal{A}_{nil}$ runs on input $\DG$ as follows: 
       at first, 
       for each $1\leq i \leq k$, identify the subgraph $\DG_i$ of $\DG$ induced on the set of vertices with out-degree $p_i$-power in $\DG$. If $\DG$ is indeed colour-isomorphic to the reduced directed power graph of some nilpotent group $G=G_1\times \dots \times G_k$, then we must have the following: each $\DDG{i}$ should correspond to the reduced directed power graph of the unique Sylow subgroup $G_i$ of $G$ (due to \Cref{claim2: patchup} in \Cref{correctness of patchup}). To check (1), we run the algorithm $\mathcal{A}$ on each $\DDG{i}$. If $\mathcal{A}$ returns `No' for some $\DDG{i}$, $1\leq i \leq k$, then $\mathcal{A}_{nil}$ returns `No'. Else, it returns $f_i:V(R(DPow(G_i))) \overrightarrow{} V(\DDG{i})$, for all $1\leq i \leq k$. Then we run the Gluing process on inputs $\DG$; $R(DPow(G)$; $R(DPow(G_i))$, $\DDG{i}, f_i$, for all $1\leq i \leq k$; . 
\end{proof}

By the preceding lemma, it suffices to design a polynomial time algorithm $\mathcal{A}$ for recognising directed power graphs within certain subclasses $\mathcal{X}$ of $\mathcal{G}_{PP}$. We focus on three such subclasses: abelian groups, groups of bounded polycyclic length, and groups of prime exponent. The last subclass is discussed in \Cref{sec: recognition nilpotent square free}. For the first two subclasses, all candidate groups can be enumerated in polynomial time (see \Cref{enumerating abelian p-groups in poly} and \Cref{bounded polycyclic length: presentation and enumerating}), followed by applying the isomorphism algorithm from \Cref{iso: nil}. The following result summarises the key outcomes of our discussion so far.


\medskip

\begin{theorem}\label{summary}
There are polynomial time algorithms $\mathcal{A}_{ab}$, $\mathcal{A}_{polycyc}$, and $\mathcal{A}_{sqfr}$ that take a graph $\Gamma$ as input and decide if $\Gamma$ is isomorphic to the power graph of an abelian group, a nilpotent group of bounded polycyclic length, or a nilpotent group of squarefree exponent, respectively.
\end{theorem}

\subsection{Abelian Groups}\label{sec: recognition abelian p-groups}




The next theorem follows as a consequence of   \Cref{enumerating abelian p-groups in poly} and \Cref{pow(nil) reducible to dpow(p)}. Subsequently, \Cref{summary} for abelian groups follows directly.

\medskip 
    
\begin{theorem}\label{recognition abelian}
There is a polynomial time algorithm that takes a reduced directed graph $\Gamma$ as input and recognises whether $\Gamma$ is isomorphic to the reduced directed power graph of an abelian group.
\end{theorem}
    

\medskip 

\begin{restatable}{lemma}{abelian}\label{enumerating abelian p-groups in poly}
    There is a polynomial time algorithm that, given a number $p^m$, outputs all non-isomorphic abelian groups of order $p^m$.
          
\end{restatable}
\begin{proof}
By the Fundamental Theorem of Abelian Groups, any abelian group of order $n=p^{m}$ 
    is isomorphic to $\ZZ_{p^{r_{1}}}\times \ZZ_{p^{r_{2}}} \times \dots \times \ZZ_{p^{r_{k}}}$, such that each $r_{i} \geq 1$, for all $1\leq i\leq k$, and $r_{1}+r_{2}+\dots + r_{k}=m$. 
    Hence, the number of non-isomorphic abelian groups of order $p^m$ is the number of partitions of $m$, which is upper bounded by $2^{c \sqrt{m}}$, for some constant $c$ (see \cite{apostol2013introduction}). Observe that, $2^{c\sqrt{m}} < (|G|)^c$. 
Let $Part(m)=\{Part_1,Part_2,\dots,Part_k\}$ be the set of all possible partitions of $m$. For each possible partition $Part_i$ in $Part(m)$, we construct the Cayley table of the corresponding abelian $p$-group $G_{Part_i}$, i.e., if $Part_i=\{m_1,m_2,\dots,m_s\}$, then $G_{Part_i}:= \ZZ_{p^{m_1}}\times \ZZ_{p^{m_2}}\times \dots \times \ZZ_{p^{m_s}}$. 
Since $|Part(m)|< (p^m)^c$, the algorithm/enumeration process runs in time polynomial in $p^m$. 
\end{proof}

\subsection{Nilpotent Groups of Bounded Polycyclic Length}\label{sec: recognition nilpotent subclass}

A \textit{polycyclic series} of a group $G$ is a finite series $G=G_0 \geq G_1 \geq \dots \geq G_l=\{1\}$ such that for each $0\leq i \leq l-1$, $G_{i+1} \lhd G_{i}$ and the factor $G_{i+1}/G_i$ is cyclic. A group with a polycyclic series is called a \textit{polycyclic group}. 
Every finite nilpotent group is a polycyclic group. The \textit{polycyclic length} (see \cite{fisher1974polycyclic})  of a polycyclic group is the number of non-trivial factors of a polycyclic series of shortest length. The main result of this section follows as a consequence of \Cref{bounded polycyclic length: presentation and enumerating} and \Cref{pow(nil) reducible to dpow(p)}. As a consequence, \Cref{summary} for nilpotent groups of bounded polycyclic length follows.


 \medskip 

 \begin{theorem}\label{recognition bounded polycyclic}
There is a polynomial time algorithm that takes a reduced directed graph $\Gamma$ as input and checks whether $\Gamma$ is isomorphic to the reduced directed power graph of a nilpotent group with bounded polycyclic length.
\end{theorem}

Let $\mathcal{X}$ denote the class of groups in $\mathcal{G}_{PP}$ with bounded polycyclic length. To enumerate these groups in polynomial time (see \Cref{bounded polycyclic length: presentation and enumerating}), we use the concept of group \emph{presentations}, which is derived here from their polycyclic series. 
Let $X$ be a finite set and $R$ is a set of words over $X\cup X^{-1}$, i.e., $R$ is a subset of $(X\cup X^{-1})^*$, where $A^*$ is the set of all strings $x_1x_2\dots x_r$ with each $x_i \in A$. The \textit{presentation} of a group $G$ is the pair $(X,R)$ in the sense that $G$ is isomorphic to the quotient group $ F_X / N$, where $F_X$ is the free group on $X$ and $N$ is the normal closure of $R$ in $F_X$. 
The length of a presentation $(X,R)$ is denoted by $||(X,R)||$.

We first show that standard methods (e.g., \cite{higman1960enumerating}, \cite{1313808}) yield logarithmic-size presentations for groups in $\mathcal{X}$. Using these, and a polynomial time \emph{procedure} for constructing a Cayley table from a presentation (cf. \cite{1313808}), we establish the next lemma. 



\medskip

\begin{restatable}{lemma}{polycyclic}\label{bounded polycyclic length: presentation and enumerating}
  Let $c>0$ be a constant and $\mathcal{X}$ be the set of groups $G$ from $\mathcal{G}_{PP}$ such that the polycyclic length of $G$ is bounded by $c$. Then, each $G\in \mathcal{X}$ has a presentation of length $O(\log |G|)$. Moreover, there is an algorithm that generates Cayley tables of all non-isomorphic groups of $\mathcal{X}$ of order $N$ in time polynomial in $N$.
\end{restatable}
\begin{proof}
We first show that $G\in \mathcal{X}$ has a logarithmic-size presentation. 

Let $|G|=p^n$. Suppose\footnote{A \textit{subgroup} $H$ of $G$ is denoted by $H \leq G$, and a \textit{normal subgroup} $H$ of $G$ is denoted by $H \lhd G$.}, $G=G_0 \geq G_1 \geq \dots \geq G_{c-1} \geq G_c=\{1\}$ is the polycyclic series of $G$, where $G_{i+1} \lhd G_{i}$, and $G_{i}/G_{i+1}$ is cyclic, for all $0\leq i \leq c-1$. We inductively construct a presentation of $(X_i,R_i)$ of $G_i$, $0\leq i \leq c-1$, where the elements of $R_i$ will be encoded as tuples of length $i$. Moreover, by induction on $i$, we show that $||R_i||\leq 2i^3 \log |G|$. 


 \textit{Base case:} For the base case, we give a representation $(X_1,R_1)$ of $G_{c-1}$. 
Let $[G_{c-1}: G_c]=r_c$ and $G_{c-1}= \langle g_{c}\rangle$. We take $X_{1}= \{g_{c}\}$. Since $G_{c-1}$ is cyclic, we have $g_{c}^{{r_c}}=1$. This relation forms the set $R_1$ and $||R_1|| \leq \log r_c < 2\log|G|$ (as $r_c \leq p^n=|G|$). Note that, to store $R_1$, storing the power of the generator $g_c$ is enough at this step.

\textit{Inductive case:} Let $(X_i,R_i)$ be a presentation for $G_{c-i}$, where $X_i=\{g_{c},g_{c-1},\dots,g_{c-i+1}\}$ and $||R_i||\leq 2i^3 \log |G|$. We now construct a presentation $(X_{i+1},R_{i+1})$ for $G_{c-i-1}$ such that $||R_{i+1}||\leq 2(i+1)^3\log |G|$. Elements of the set $R_i$ are represented as tuples of length $i$, encoding the exponents of the generators $g_c, g_{c-1}, \ldots, g_{c-i+1}$ \emph{in that order}. Accordingly, elements in $R_{i+1} \setminus R_i$ are tuples of length $i+1$, with the final component representing the exponent of the newly introduced generator $g_{c-i}$ in $X_{i+1}$.

    Since $G_{c-i-1}/G_{c-i}$ is cyclic, let $[G_{c-i-1}:G_{c-i}] =r_{c-i}$ and $G_{c-i-1}/G_{c-i}= \langle g_{c-i}G_{c-i} \rangle$. We set $X_{i+1}=X_i \cup \{g_{c-i}\}$. Now we define the relations of $R_{i+1}\setminus R_{i}$. 

    Since $G_{c-i}$ is normal in $G_{c-i-1}$, we have $g_{c-i}G_{c-i}=G_{c-i}g_{c-i}$. Therefore, we have the following set of relations:
    \begin{equation}
    \label{label 1}
        g_{c-i}g_{r}=w_{c-i,r}g_{c-i}, \ c \geq r \geq c-i+1\\
        \text{ for some $w_{c-i,r} \in G_{c-i}$}
    \end{equation}
  Note that here $w_{c-i,r}$ is an word of the form $g_{c}^{l_{c}}g_{c-1}^{l_{c-1}} \dots g_{c-i+1}^{l_{c-i+1}}$ where $0 \leq l_{u} \leq r_{u}$, for each $ c-i+1 \leq u \leq c$. Since $X_i$ is a generating set for $G_{c-i}$ and the order of $g_{u} \in X_i$ is $r_u$, any arbitrary element of $G_{c-i}$ can be represented by an word of the form $w_{c-i,r}$. 
    There are $i$ relations of the form \Cref{label 1}, and the length of the presentation for each relation is at most $i.n. \log p$.

    Now since $G_{c-i-1}/G_{c-i} \cong \ZZ_{p^{r_{c-i}}}$, we have the following relation:
    \begin{equation} \label{label 3}
        g_{c-i}^{p^{r_{c-i}}}=w''_{c-i}, \text{for some $w''_i \in Z_i$}
    \end{equation}

    Here also $w''_{i}$ is a word defining element of $Z_i$. Note that the above relation has presentation of length at most $i. n. \log p$.

    We define $R_{i+1}$ as the union of the above relations \Cref{label 1}, \Cref{label 3} and $R_i$. 
Since $\log (|G|)=n .  \log p$, we have the following: $||R_{i+1}||\leq ||R_i|| + (i+1). i. \log |G| \leq ||R_i|| + 2i^2\log (|G|)$, for $i \geq 1$. By the induction hypothesis, this implies $||R_{i+1}||\leq 2(i+1)^3\log(|G|)$.

By induction, the length of the presentation for $(X_c,R_c)$, which defines $G_0=G$, is at most $2c^3\log (|G|)$. Hence, $G$ has a presentation of length $O(\log |G|)$. This proves the first part of the lemma. 

We now prove the second part of the lemma.

Arvind and Toran \cite{1313808} showed that there is a polynomial time algorithm that takes a presentation $(X,R)$ of a solvable group $G$ which is obtained from a process described in the proof of the above lemma and outputs the Cayley table of a group $G'$ isomorphic to $G$. 
On all strings of length at most $O(\log N)$, we run the algorithm due to Arvind and Toran and generate the Cayley tables of all non-isomorphic groups of order $N$. The process of generating a presentation of length $\log(|G|)$ of each $G\in \mathcal{X}$ guarantees that this works.
\end{proof}

A number $n\in \mathbb{N}$ is \textit{$d$-powerfree} if for each prime $p$, $p^d$ does not divide $n$. Note that a $p$-group of order $p^{d-1}$ has polycyclic length at most $(d-1)$. The next corollary follows from \Cref{summary} of nilpotent groups of bounded polycyclic length.

\medskip 

\begin{corollary}\label{bounded power}
Let $d$ be a fixed number. Let $ n \in \mathbb{N} $ be $ d$-powerfree . Then, there exists a polynomial time algorithm that determines whether an undirected graph is isomorphic to the power graph of a nilpotent group of order $n$.
\end{corollary}

\subsection{Nilpotent Groups of Squarefree Exponent}\label{sec: recognition nilpotent square free}

We begin with a lemma that follows from a result by \cite{pourgholi2015undirected}.

\medskip 

\begin{lemma}\label{exponet p} Let $p$ be a prime. A reduced directed graph $\DG$ is the reduced directed power graph of a finite $p$-group of order $p^n$ of exponent $p$ if and only if $\DG$ is colour-isomorphic to a directed star having $(p^n-1)/(p-1)$ vertices each with colour $p$, out-degree 1 and in-degree 0, and one vertex with colour $1$, out-degree 0 and in-degree $(p^n-1)/(p-1)$.
\end{lemma}

From the above lemma 
and \Cref{pow(nil) reducible to dpow(p)}, the next result can be easily observed by considering the group class $\mathcal{X}=\{G \in \mathcal{G}_{PP} :\text{ exponent of $G$ is } p\}$.

\medskip

\begin{theorem}\label{recognition squarefree exponent}
There is a polynomial time algorithm that takes a reduced directed graph $\Gamma$ as input, checks whether $\Gamma$ is isomorphic to the reduced directed power graph of a nilpotent group with squarefree exponent.
\end{theorem}

From the above theorem, one can easily observe that \Cref{summary} for nilpotent groups of squarefree exponent holds.

\printbibliography 
\begin{appendix}
\renewcommand{\theequation}{\arabic{equation}}
\setcounter{equation}{0}



\section{Appendix}
\markboth{Appendix}{Appendix}

\subsection{Proof of \Cref{sat to star-clique} of \Cref{sat to pow}}\label{app: proof of claim 1(sat to star-clique)}

    Suppose that ${\Phi}$ has a satisfying assignment, in which each clause $C_j$, $1\leq j \leq m$, of ${\Phi}$ has at least one true literal. We select the vertices corresponding to the true literals in the satisfying assignment. Let $L$ be the set of selected vertices. 
    Then $L \cup \{r\}$ induces a connected subgraph in $\GG$ and is coloured with $\{0,1,2,\dots,n\}$. Since each clause $C_i$ is true by the satisfying assignment of $\Phi$, the construction of $\GG$ implies the existence of $Z \subseteq N(L)\setminus \{r\}$ such that $col(Z)=\{C_1,C_2,\dots,C_m\}$. Therefore, $L\cup \{r\} \cup Z$ is an occurrence of $M$ in $\GG$.

    For the other direction, suppose that $\GG$ has an occurrence $V'$ of $M$. We want to show that ${\Phi}$ is satisfiable. 
    Let $L \subseteq V'$ and $Z \subseteq V'$ be the subsets coloured with $\{1,2,\dots,n\}$ and $\{C_1,C_2,\dots, C_m\}$ respectively. We set literals $l_i$ corresponding to each vertex in $L$ as true literals. Since $V'$ induces a connected subgraph, each vertex $z_{i_j} \in Z$ coloured with $C_{i_j}$ is adjacent to some vertex $x_i$ (or $\overline{x}_i$) in $L$. Moreover, $Z$ is coloured with $\{C_1,C_2,\dots,C_m\}$. Hence, our assignment process satisfies each clause $C_j$, $1 \leq j \leq m$, and thus $\Phi$. 

\subsection{The Issue Behind the Direct Application of the Isomorphism Algorithm by Das et al. in Recognising the Power Graphs of Nilpotent Groups (\Cref{sec: recognition})}\label{story behind the two techniques}

The algorithm due to Das et al. \cite{das_et_al:LIPIcs.FSTTCS.2024.20} works by noting that for two nilpotent groups $G^{(1)}$ and $G^{(2)}$ of order $n$, $DPow(G^{(1)}) \cong DPow(G^{(2)})$ if and only if for each prime $p$ which divides $n$, $DPow(S_p^{(1)})\cong DPow(S_p^{(2)})$, where $S_p^{(1)}$ and $S_p^{(2)}$ are the Sylow $p$-subgroups of $G^{(1)}$ and $G^{(2)}$, respectively. They use this information for their isomorphism algorithm of directed power graphs of finite nilpotent groups. Given a directed power graph $\Gamma$ of a nilpotent $G$, it is easy to find the directed power graph corresponding to its Sylow $p$-subgroup: just take the directed subgraph induced by the vertices of out-degree divisible by $p$. The isomorphism algorithm by \cite{das_et_al:LIPIcs.FSTTCS.2024.20} primarily checks the isomorphism of these induced subgraphs. However, an input directed graph $\DG$ may not be a directed power graph, but the subgraphs induced on the set of prime-power out-degree vertices may still be valid directed power graphs of $p$-groups. This is the reason why the algorithm by Das et al. \cite{das_et_al:LIPIcs.FSTTCS.2024.20} may still declare that the input graphs are isomorphic while they are not, yielding a false positive answer. An illustrative example is provided in \Cref{Gluing Illustration}. For a compact representation, we have used certain notations and concepts (reduced graphs) that are mentioned in the main body.

\begin{figure}[hpt!]
\centering

\begin{subfigure}{0.43\textwidth}
        \begin{tikzpicture}[scale=0.3]
        
        \coordinate (a) at (7,1);
        \coordinate (g1) at (2,5);
        \coordinate (g2) at (4,5);
        \coordinate (g3) at (6,5);
        
        \coordinate (g4) at (1,9);
        \coordinate (g5) at (3.5,9);
        
        \coordinate (g6) at (0,13);
        \coordinate (g7) at (3,13);
        
        \coordinate (h1) at (9,5);
        \coordinate (h2) at (11,5);
        \coordinate (h3) at (13,5);
        \coordinate (h4) at (15,5);
        
        \coordinate (m1) at (7,12.5);
        \coordinate (m2) at (9,12.5);
        \coordinate (m3) at (11,12.5);
        \coordinate (m4) at (13,12.5);
        
        \coordinate (n1) at (0,18);
        \coordinate (n2) at (2,18);
        \coordinate (n3) at (4,18);
        \coordinate (n4) at (6,18);
        \coordinate (n5) at (8,18);
        \coordinate (n6) at (10,18);
        \coordinate (n7) at (12,18);
        \coordinate (n8) at (14,18);

        \tikzset{
  middlearrow/.style={
    postaction={
      decorate,
      decoration={
        markings,
        mark=at position 0.5 with {\arrow{>}}
      }
    }
  }
}

        \tikzset{
  quatarrow/.style={
    postaction={
      decorate,
      decoration={
        markings,
        mark=at position 0.2 with {\arrow{>}}
      }
    }
  }
}

        \tikzset{
  quaterarrow/.style={
    postaction={
      decorate,
      decoration={
        markings,
        mark=at position 0.450 with {\arrow{>}}
      }
    }
  }
}

         \path [draw, middlearrow] (g1) -- (a);
         \path [draw, middlearrow] (g2) -- (a);
        \path [draw, middlearrow] (g3) -- (a);
        \path [draw, middlearrow] (g4) -- (g1);
        \path [draw, middlearrow] (g5) -- (g1);
        \path [draw, middlearrow] (g6) -- (g4);
        \path [draw, middlearrow] (g7) -- (g4);
        \path [draw, middlearrow] (g4) -- (a);
        \path [draw, middlearrow] (g5) -- (a);
        \path [draw, middlearrow] (g6) -- (a);
        \path [draw, middlearrow] (g7) -- (a);
        
       \path [draw, middlearrow] (g6) .. controls (-2,9) .. (g1);
       \path [draw, quatarrow] (g7) -- (g1);        \path [draw, middlearrow] (h1) -- (a);
        \path [draw, middlearrow] (h2) -- (a);
        \path [draw, middlearrow] (h3) -- (a);
        \path [draw, middlearrow] (h4) -- (a);
        \path [draw, quatarrow] (n1) -- (m1);
        \path [draw, quatarrow] (n2) -- (m2);
        \path [draw, quatarrow] (n3) -- (m3);
        \path [draw, quatarrow] (n4) -- (m4);
        \path [draw, quaterarrow] (n5) -- (m1);
        \path [draw, quaterarrow] (n6) -- (m2);
        \path [draw, quaterarrow] (n7) -- (m3);
        \path [draw, quaterarrow] (n8) -- (m4);

     \filldraw [brown] (a) circle(10pt);
        \filldraw [blue] (g1) circle(10pt);
        \filldraw [blue] (g2) circle(10pt);
        \filldraw [blue] (g3) circle(10pt);
        \filldraw [blue] (g4) circle(10pt);
        \filldraw [blue] (g5) circle(10pt);
        \filldraw [blue] (g6) circle(10pt);
        \filldraw [blue] (g7) circle(10pt);
        
        \filldraw [cyan] (h1) circle(10pt);
        \filldraw [cyan] (h2) circle(10pt);
        \filldraw [cyan] (h3) circle(10pt);
        \filldraw [cyan] (h4) circle(10pt);
        
        \filldraw [teal] (m1) circle(10pt);
        \filldraw [teal] (m2) circle(10pt);
        \filldraw [teal] (m3) circle(10pt);
        \filldraw [teal] (m4) circle(10pt);
        
        \filldraw [teal] (n1) circle(10pt);
        \filldraw [teal] (n2) circle(10pt);
        \filldraw [teal] (n3) circle(10pt);
        \filldraw [teal] (n4) circle(10pt);
        \filldraw [teal] (n5) circle(10pt);
        \filldraw [teal] (n6) circle(10pt);
        \filldraw [teal] (n7) circle(10pt);
        \filldraw [teal] (n8) circle(10pt);

        \draw (g1) node[below, yshift=-2pt, scale=0.7]{$g_1$};
        \draw (g2) node[below, yshift=-2pt, scale=0.7]{$g_2$};
        \draw (g3) node[right, xshift=2pt, scale=0.7]{$g_3$};
        \draw (g4) node[left, xshift=-2pt, scale=0.7]{$g_4$};
        \draw (g5) node[above, yshift=2pt, scale=0.7]{$g_5$};
        \draw (g6) node[above, yshift=2pt, scale=0.7]{$g_6$};
        \draw (g7) node[above, yshift=2pt, scale=0.7]{$g_7$};
        
        \draw (h1) node[above, yshift=2pt, scale=0.7]{$h_1$};
        \draw (h2) node[above, yshift=2pt, scale=0.7]{$h_2$};
        \draw (h3) node[above, yshift=2pt, scale=0.7]{$h_3$};
        \draw (h4) node[above, yshift=2pt, scale=0.7]{$h_4$};
        
        \draw (m1) node[below, xshift=-10pt, yshift=-2pt, scale=0.7]{$(g_4,h_1)$};
        \draw (m2) node[below, yshift=-2pt, scale=0.7]{$(g_4,h_2)$};
        \draw (m3) node[below, xshift=4pt , yshift=-2pt, scale=0.7]{$(g_4,h_3)$};
        \draw (m4) node[below, xshift=10pt, yshift=-2pt, scale=0.7]{$(g_4,h_4)$};
        
                \draw (n1) node[above, xshift=-12pt, yshift=2pt, scale=0.65]{$(g_6,h_1)$};
        \draw (n2) node[above, xshift=-7 pt, yshift=2pt, scale=0.65]{$(g_6,h_2)$};
        \draw (n3) node[above, xshift=-4pt , yshift=2pt, scale=0.65]{$(g_6,h_3)$};
        \draw (n4) node[above, xshift=-2pt, yshift=2pt, scale=0.65]{$(g_6,h_4)$};
        
                        \draw (n5) node[above, xshift=2pt, yshift=2pt, scale=0.65]{$(g_7,h_1)$};
        \draw (n6) node[above, xshift=4 pt, yshift=2pt, scale=0.65]{$(g_7,h_2)$};
        \draw (n7) node[above, xshift=7pt , yshift=2pt, scale=0.65]{$(g_7,h_3)$};
        \draw (n8) node[above, xshift=12pt, yshift=2pt, scale=0.65]{$(g_7,h_4)$};

         \end{tikzpicture}
        \caption{ $R(DPow(H)$, where $H= \mathbb{Z}_8\times\mathbb{Z}_2 \times \mathbb{Z}_3 \times \mathbb{Z}_3$}
\end{subfigure}
\hspace{0.5 cm}
\begin{subfigure}{0.43\textwidth}
\centering
        \begin{tikzpicture}[scale=0.3]

        \coordinate (a) at (7,1);
        \coordinate (g1) at (2,5);
        \coordinate (g2) at (4,5);
        \coordinate (g3) at (6,5);
        
        \coordinate (g4) at (1,9);
        \coordinate (g5) at (3.5,9);
        
        \coordinate (g6) at (0,13);
        \coordinate (g7) at (3,13);
        
        \coordinate (h1) at (9,5);
        \coordinate (h2) at (11,5);
        \coordinate (h3) at (13,5);
        \coordinate (h4) at (15,5);
        
        \coordinate (m1) at (7,12.5);
        \coordinate (m2) at (9,12.5);
        \coordinate (m3) at (11,12.5);
        \coordinate (m4) at (13,12.5);
        
        \coordinate (n1) at (0,18);
        \coordinate (n2) at (2,18);
        \coordinate (n3) at (4,18);
        \coordinate (n4) at (6,18);
        \coordinate (n5) at (8,18);
        \coordinate (n6) at (10,18);
        \coordinate (n7) at (12,18);
        \coordinate (n8) at (14,18);
        
        \filldraw [brown] (a) circle(10pt);
        \filldraw [blue] (g1) circle(10pt);
        \filldraw [blue] (g2) circle(10pt);
        \filldraw [blue] (g3) circle(10pt);
        \filldraw [blue] (g4) circle(10pt);
        \filldraw [blue] (g5) circle(10pt);
        \filldraw [blue] (g6) circle(10pt);
        \filldraw [blue] (g7) circle(10pt);
        
        \filldraw [cyan] (h1) circle(10pt);
        \filldraw [cyan] (h2) circle(10pt);
        \filldraw [cyan] (h3) circle(10pt);
        \filldraw [cyan] (h4) circle(10pt);
        
        \filldraw [teal] (m1) circle(10pt);
        \filldraw [teal] (m2) circle(10pt);
        \filldraw [teal] (m3) circle(10pt);
        \filldraw [teal] (m4) circle(10pt);
        
        \filldraw [teal] (n1) circle(10pt);
        \filldraw [teal] (n2) circle(10pt);
        \filldraw [teal] (n3) circle(10pt);
        \filldraw [teal] (n4) circle(10pt);
        \filldraw [teal] (n5) circle(10pt);
        \filldraw [teal] (n6) circle(10pt);
        \filldraw [teal] (n7) circle(10pt);
        \filldraw [teal] (n8) circle(10pt);

        \tikzset{
  middlearrow/.style={
    postaction={
      decorate,
      decoration={
        markings,
        mark=at position 0.5 with {\arrow{>}}
      }
    }
  }
}

        \tikzset{
  quatarrow/.style={
    postaction={
      decorate,
      decoration={
        markings,
        mark=at position 0.2 with {\arrow{>}}
      }
    }
  }
}

        \tikzset{
  quaterarrow/.style={
    postaction={
      decorate,
      decoration={
        markings,
        mark=at position 0.450 with {\arrow{>}}
      }
    }
  }
}

         \path [draw, middlearrow] (g1) -- (a);
         \path [draw, middlearrow] (g2) -- (a);
        \path [draw, middlearrow] (g3) -- (a);
        \path [draw, middlearrow] (g4) -- (g1);
        \path [draw, middlearrow] (g5) -- (g1);
        \path [draw, middlearrow] (g6) -- (g4);
        \path [draw, middlearrow] (g7) -- (g4);
        \path [draw, middlearrow] (g4) -- (a);
        \path [draw, middlearrow] (g5) -- (a);
        \path [draw, middlearrow] (g6) -- (a);
        \path [draw, middlearrow] (g7) -- (a);
        
       \path [draw, middlearrow] (g6) .. controls (-2,9) .. (g1);
       \path [draw, quatarrow] (g7) -- (g1);
        
        \path [draw, middlearrow] (h1) -- (a);
        \path [draw, middlearrow] (h2) -- (a);
        \path [draw, middlearrow] (h3) -- (a);
        \path [draw, middlearrow] (h4) -- (a);
        \path [draw, quatarrow] (n1) -- (m1);
        \path [draw, quatarrow] (n2) -- (m2);
        \path [draw, quatarrow] (n3) -- (m3);
        \path [draw, quatarrow] (n4) -- (m4);
        \path [draw, quaterarrow] (n5) -- (m1);
        \path [draw, quaterarrow] (n6) -- (m2);
        \path [draw, quaterarrow] (n7) -- (m4);
        \path [draw, quaterarrow] (n8) -- (m4);
        
        \filldraw [brown] (a) circle(10pt);
        \filldraw [blue] (g1) circle(10pt);
        \filldraw [blue] (g2) circle(10pt);
        \filldraw [blue] (g3) circle(10pt);
        \filldraw [blue] (g4) circle(10pt);
        \filldraw [blue] (g5) circle(10pt);
        \filldraw [blue] (g6) circle(10pt);
        \filldraw [blue] (g7) circle(10pt);
        
        \filldraw [cyan] (h1) circle(10pt);
        \filldraw [cyan] (h2) circle(10pt);
        \filldraw [cyan] (h3) circle(10pt);
        \filldraw [cyan] (h4) circle(10pt);
        
        \filldraw [teal] (m1) circle(10pt);
        \filldraw [teal] (m2) circle(10pt);
        \filldraw [teal] (m3) circle(10pt);
        \filldraw [teal] (m4) circle(10pt);
        
        \filldraw [teal] (n1) circle(10pt);
        \filldraw [teal] (n2) circle(10pt);
        \filldraw [teal] (n3) circle(10pt);
        \filldraw [teal] (n4) circle(10pt);
        \filldraw [teal] (n5) circle(10pt);
        \filldraw [teal] (n6) circle(10pt);
        \filldraw [teal] (n7) circle(10pt);
        \filldraw [teal] (n8) circle(10pt);

                \draw (g1) node[below, yshift=-2pt, scale=0.7]{$g_1$};
        \draw (g2) node[below, yshift=-2pt, scale=0.7]{$g_2$};
        \draw (g3) node[right, xshift=2pt, scale=0.7]{$g_3$};
        \draw (g4) node[left, xshift=-2pt, scale=0.7]{$g_4$};
        \draw (g5) node[above, yshift=2pt, scale=0.7]{$g_5$};
        \draw (g6) node[above, yshift=2pt, scale=0.7]{$g_6$};
        \draw (g7) node[above, yshift=2pt, scale=0.7]{$g_7$};
        
        \draw (h1) node[above, yshift=2pt, scale=0.7]{$h_1$};
        \draw (h2) node[above, yshift=2pt, scale=0.7]{$h_2$};
        \draw (h3) node[above, yshift=2pt, scale=0.7]{$h_3$};
        \draw (h4) node[above, yshift=2pt, scale=0.7]{$h_4$};
        
                \draw (m1) node[below, xshift=-10pt, yshift=-2pt, scale=0.7]{$(g_4,h_1)$};
        \draw (m2) node[below, yshift=-2pt, scale=0.7]{$(g_4,h_2)$};
        \draw (m3) node[below, xshift=4pt , yshift=-2pt, scale=0.7]{$(g_4,h_3)$};
        \draw (m4) node[below, xshift=10pt, yshift=-2pt, scale=0.7]{$(g_4,h_4)$};
        
                \draw (n1) node[above, xshift=-12pt, yshift=2pt, scale=0.65]{$(g_6,h_1)$};
        \draw (n2) node[above, xshift=-7 pt, yshift=2pt, scale=0.65]{$(g_6,h_2)$};
        \draw (n3) node[above, xshift=-4pt , yshift=2pt, scale=0.65]{$(g_6,h_3)$};
        \draw (n4) node[above, xshift=-2pt, yshift=2pt, scale=0.65]{$(g_6,h_4)$};
        
                        \draw (n5) node[above, xshift=2pt, yshift=2pt, scale=0.65]{$(g_7,h_1)$};
        \draw (n6) node[above, xshift=4 pt, yshift=2pt, scale=0.65]{$(g_7,h_2)$};
        \draw (n7) node[above, xshift=7pt , yshift=2pt, scale=0.65]{$(g_7,h_3)$};
        \draw (n8) node[above, xshift=12pt, yshift=2pt, scale=0.65]{$(g_7,h_4)$};
        	    
          \end{tikzpicture}
        \caption{The given $\DDG{1}$}
\end{subfigure}
\caption{A short representation of $R(DPow(H))$ and $\DDG{1}$. Some edges and vertices are omitted for a clear representation.}
\label{Gluing Illustration}
\end{figure}
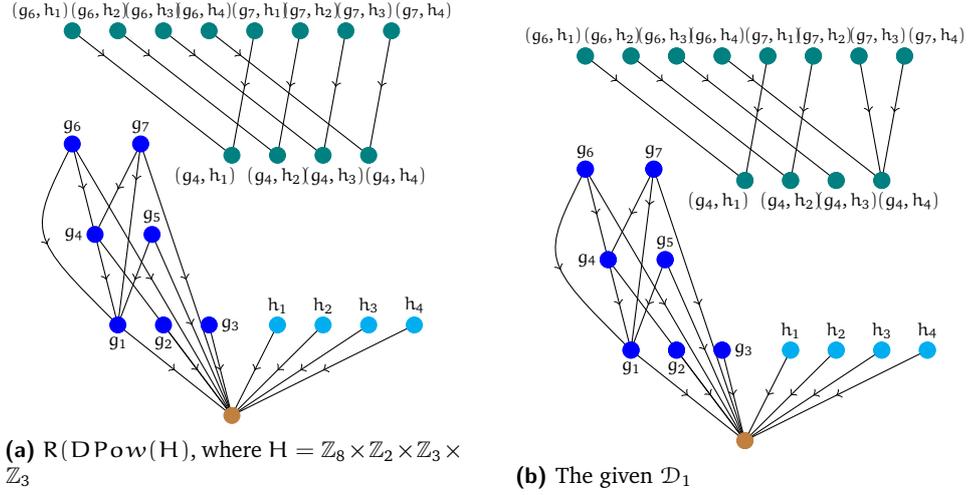

\textit{Description of the example depicted in \Cref{Gluing Illustration}:} Here $\DDG{1}$ is a reduced directed graph and $R(DPow(H))$ is the reduced directed power graph of $H=\mathbb{Z}_8 \times \mathbb{Z}_2 \times \mathbb{Z}_3 \times \mathbb{Z}_3$. Let $\DDG{11}$ and $\DDG{12}$ denote the subgraphs of $\DDG{1}$ induced on the set of vertices with $2$-power and $3$-power colours. This example shows an instance where $\DDG{11}\cong_c R(DPow(H_1))$ and $\DDG{12}\cong_c R(DPow(H_2))$, but $\DDG{1} \ncong_c R(DPow(H))$.

 For a fixed $i$ and $j$, the same notation $g_i$, $h_j$ and $(g_i,h_j)$ are used to denote two isomorphic vertices in the graphs $R(DPow(H))$ and $\DDG{1}$. Observe that, $col(g_4)=col(g_5)=4$ and $col(g_6)=col(g_7)=8$, whereas $col(h_1)=col(h_2)=col(h_3)=col(h_4)=3$. Hence, to preserve colours, $col((g_6,h_j))=col(g_7,h_j))=24$ and $col((g_4,h_1))=col((g_5,h_1))=12$ in both the graphs. But note that the subgraphs induced on the set of these $(g_i,h_j)$, $i=4,\dots,7$, vertices (denoted by green vertices in the picture) do not follow the adjacency criteria.
 


\subsection{Proof of Correctness of Gluing Process (\Cref{sec: recognition})}\label{correctness of patchup}

Here, we prove the correctness of Gluing process described in \Cref{sec: recognition}. Let $\DDG{1}$ be a reduced directed graph and $H$ be a finite nilpotent group such that $H=H_1\times \dots \times H_k$, where $H_i$ is a Sylow $p_i$-subgroup. We want to prove that if $\DG_1 \cong_c \DG_2 $, where $R(DPow(H))=\DG_2$, the mapping $f$ defined for Gluing process is indeed a colour-preserving isomorphism (see \Cref{rong mile geche}). We start with some discussions which facilitate our proof.

Since $\DG_1 \cong_c \DG_2 $, we can assume that there exists some nilpotent group $G=G_1 \times G_2 \times \dots G_k$, where $G_i$ is the Sylow $p_i$-subgroup such that $\DG_1 = R(DPow(G))$. Note that the groups $G$ and $H$ may be non-isomorphic, but $|G|=|H|$. For all $1\leq i< j \leq k$, $gcd(|H_i|,|H_j|)=1, gcd(|G_i|,|G_j|)=1$, and for all $1\leq i \leq k$, $|H_i|=|G_i|$. Each element of $H$ can be expressed as $(h_1,h_2,\dots,h_k)$, where $h_i \in H_i$, and each $h_i$ can also be identified as an element of $H$ by embedding it in the tuple $(e,\dots,e,h_i,e,\dots,e)$. So a vertex of the reduced graph $\DDG{2}$ is represented as $[(h_1,h_2,\dots,h_k)]$. Similar things apply for $G$ and its associated graph $\DDG{1}$.

\medskip 

\begin{remark}\label{nilpotent remark}
   Let $(x_1,x_2,\dots,x_k)$ and $(y_1,y_2,\dots,y_k)$ be two elements of a finite nilpotent group $G_1\times G_2 \times \dots \times G_k$. Then the following two hold:
   
   (a) $[(x_1,x_2,\dots,x_k)]=[(y_1,y_2,\dots,y_k)]$ if and only if $[x_i]=[y_i]$, for all $1\leq i \leq k$. This can be argued easily by the fact that $ord(x_i)$ and $ord(x_j)$ are co-prime to each other, for all $1\leq i < j \leq k$, and $ord(x_i)=ord(y_i)$, for all $1 \leq i \leq k$.

   (b) $(x_1,x_2,\dots,x_k)$ generates $(y_1,y_2,\dots,y_k)$ if and only if, for all $1\leq i \leq k$, $x_i$ generates $y_i$ in $G_i$ (or, in $\{e\}\times \dots \times \{e\} \times G_i \times \{e\} \times \dots \times \{e\}$ ).
\end{remark}

Let $\DDG{ji}$ where $j=1,2$ and $i=1,2,\dots,k$ denote the subgraph of $\DDG{j}$ induced on the set of vertices with $p_i$-power colours. We give sketches for the proofs of the following two easy lemmas.

\medskip 

\begin{lemma}\label{claim1: patchup}
    $\DDG{1i} \cong_c R(DPow(\{e\}\times \dots \times\{e\}\times G_i \times \{e\} \times \dots \times \{e\}))$ and $\DDG{2i} \cong_c R(DPow(\{e\}\times \dots \times\{e\}\times H_i \times \{e\} \times \dots \times \{e\}))$, for all $1\leq i \leq k$.
\end{lemma}

\begin{proof}
    Note that, $\DG_1 \cong_c R(DPow(G))$ where $G=G_1 \times G_2 \times \dots G_k$. Here, $G_i$ is the Sylow $p_i$-subgroup of $G$. So, $|G_i|$ and $|G_j|$ are relatively prime when $i \neq j$. Also, the colour of an element $[g]$ in $\DG_1$ is the same as $ord(g)$, the order of the element $g$. Hence, all the elements of $\DDG{1}$ with $p_i$-power colours are essentially the elements from $ \{e\}\times \dots \times\{e\}\times G_i \times \{e\} \times \dots \times \{e\}$. This proves that $\DDG{1i} \cong_c R(DPow(\{e\}\times \dots \times\{e\}\times G_i \times \{e\} \times \dots \times \{e\}))$.
\end{proof}

\begin{lemma}\label{claim2: patchup}
    Let $\DDG{1} \cong_c \DDG{2}$. Then $\DDG{1i} \cong_c \DDG{2i}$, for all $1\leq i \leq k$.
\end{lemma}
\begin{proof}
    Let $\psi$ be a colour isomorphism from $\DDG{1}$ to $\DDG{2}$. Observe that all the elements of $\DDG{1i}$ can have their images only in $\DDG{2i}$ under $\psi$, since for $j\neq i$ the colours of elements in $\DDG{2j}$ are co-prime to the colours of elements in $\DDG{2i}$. So, $\psi$ restricted on $\DDG{1i}$ gives a colour-preserving isomorphism from $\DDG{1i}$ to $\DDG{2i}$,  for all $1 \leq i \leq k$.
\end{proof}

We first state a relevant fact, which follows from the properties of finite nilpotent groups, e.g, \Cref{nilpotent remark}.

\medskip

\begin{fact}\label{nil unique lca}
    Let $G=G_1 \times G_2 \times \dots \times G_k$ be a finite nilpotent group, where $G_i$ denotes the Sylow $p_i$-subgroup of $G$. Let $u_i \in G_i$, for all $1\leq i \leq k$. Then, $[(u_1,u_2,\dots,u_k)]$ is the unique least common parent of the vertices $[(u_1,e,\dots,e)],[(e,u_2,e,\dots,e)],$  $\dots,[(e,\dots,e,u_k)]$ in $R(DPow(G))$.
\end{fact}

The next lemma implies that if a reduced directed graph $\DDG{1}$ is colour-isomorphic to the reduced directed power graph of a nilpotent group $H_1\times H_2\times \dots \times H_k$, then the mapping $f$ for Gluing process is indeed a colour-preserving isomorphism between $\DDG{1}$ and $R(DPow(H_1\times H_2\times \dots \times H_k))$. This is enough for the correctness of Gluing Process, as the opposite direction is obvious, i.e, if $f$ is a colour-preserving isomorphism between $\DDG{1}$ and $R(DPow(H_1\times H_2\times \dots \times H_k))$, then these two directed graphs are indeed colour-isomorphic. At this point, let us recall the definition of $f: V(R(DPow(H))) \xrightarrow{} V(\DDG{1})$: Let $f_i$ be a colour isomorphism from $R(DPow(H_i))$ to $\DDG{1i}$, for all $1\leq i \leq k$. The function $f$ maps $[(x_1,x_2,\dots,x_k)]$ to a least common parent (if it exists) of $f_1([x_1])$, $f_2([x_2])$, $\dots$ and $f_k([x_k])$ in $\DDG{1}$. By our previous discussion, $\DDG{1}$ can be viewed as $\DG_1 = R(DPow(G))$, for some nilpotent group $G$.
 Hence, using \Cref{nil unique lca},  a vertex $[(z_1,z_2,\dots,z_k)]$ in $\DDG{1}$ is the unique least common parent of $[z_1]=[(z_1,e,\dots,e)]$, $[z_2]=[(e,z_2,\dots,e)]$, \dots, $[z_k]=[(e,e,\dots,z_k)]$. 

 \medskip 

\begin{lemma}
\label{rong mile geche}
    Let $\DDG{1}=R(DPow(G_1\times\ldots\times G_k ))$ and $\DDG{2}=R(DPow(H_1\times \dots \times H_k))$,  where $H_i$ and $G_i$ are $p_i$-groups for distinct primes $p_1,\ldots, p_k$. Let $f_i:V(\DDG{2i})\xrightarrow{} V(\DDG{1i})$ be a colour-preserving isomorphism between $\DDG{2i}$ and $\DDG{1i}$, for all $1 \leq i \leq k$. Then there exists a colour-preserving isomorphism $f:V(\DDG{2}) \xrightarrow[]{} V(\DDG{1})$ such that $f|_{\DDG{2i}}=f_i$. 
\end{lemma}
\begin{proof}
    We define mapping $f:V(\DDG{2}) \xrightarrow[]{} V(\DDG{1})$ as follows:  $f([(x_1,\dots,x_k)])= [(z_1,\dots,z_k)]$ where $f_i([x_i])=[z_i]$, for $1\leq i \leq k$.

    $(i)$ We need to prove that $f$ is well-defined. That means if $[(x_1,x_2,\dots,x_k)]=[(y_1,y_2,\dots,y_k)]$, then $f([(x_1,x_2,\dots,x_k)])=f([(y_1,y_2,\dots,y_k)])$. Note that, by \Cref{nilpotent remark}(a), $[(x_1,x_2,\dots,x_k)]=$\newline $[(y_1,y_2,\dots,y_k)]$ implies $[x_i]=[y_i]$, for all $1 \leq i \leq k$, and hence $f_i([x_i])=f_i([y_i])$ (by well-definedness of $f_i$), for all $1\leq i \leq k$. So, again by \Cref{nilpotent remark}(a), $f([(x_1,x_2,\dots,x_k)])=f([(y_1,y_2,\dots,y_k)])$.

    $(ii)$ Now, we prove that $f$ is injective and surjective. Again, by \Cref{nilpotent remark}(a) and the injective property of $f_i$, $1\leq i \leq k$, it is easy to see that $f$ is injective. To show $f$ is surjective, let $[(z_1,z_2,\dots,z_k)] \in V(\DDG{1})$. Hence, there exists $[x_i] \in V(\DDG{2i})$ such that $f_i([x_i])=[z_i]$, for all $1\leq i \leq k$ (since $f_i$ is bijection). So, by definition, $f([(x_1,x_2,\dots,x_k)])=[(z_1,z_2,\dots,z_k)]$.

    $(iii)$ We now prove that $f$ is a colour-preserving isomorphism by showing that it preserves both colour and the adjacency relation between $\DDG{2}$ and $\DDG{1}$. Let there be a directed edge from $[(x_1,x_2,\dots,x_k)]$ to $[(y_1,y_2,\dots,y_k)]$ in $\DDG{2}$. This implies that 
    $(x_1,x_2,\dots,x_k)$  generates $(y_1,y_2,\dots,y_k)$ in $H$. 
    Hence, by \Cref{nilpotent remark}(b), $x_i$ generates $y_i$ in $H_i$, for all $1\leq i \leq k$. So, there is an edge from $[x_i]$ to $[y_i]$ in $\DDG{2i}$, for all $i$. Since $f_i$ is an isomorphism between $\DDG{2i}$ and $\DDG{1i}$, there is a directed edge from $f_i([x_i])$ to $f_i([y_i])$ in $\DDG{1i}$, for all $i$. So, assuming $f_i([x_i])=[z_i]$ and $f_i([y_i])=[w_i]$, we have $z_i$ generates $w_i$ in $G_i$ for all $i$. Therefore, by \Cref{nilpotent remark}(b), there is a directed edge from $[(z_1,z_2,\dots,z_k)]$ to $[(w_1,w_2,\dots,w_k)]$ in $\DDG{1}$.

    The other direction can be similarly shown.

    Now we prove that $f$ preserves colour. Note that $ord((x_1,x_2,\dots,x_k))=ord(x_1)\times ord(x_2) \times \dots \times ord(x_k)$. Hence, $col([(x_1,x_2,\dots,x_k)])=col([x_1])\times \dots \times col([x_k])$. Since all $f_i$'s are colour-preserving, $f$ preserves colour.

    $(iv)$ We prove $f|_{\DDG{2i}}=f_i$ by showing that if $f([(x_1,x_2,\dots , x_i,\dots,x_k)])=[(z_1,z_2,\dots, z_i, \dots, z_k)]$ then $f([(e,\dots,x_i,\dots,e)])=[(e,\dots,z_i,\dots,e)]$. By the property of colour-preserving isomorphism, $col([(x_1,x_2,\dots,x_k)])$ $=col([(z_1,z_2,\dots,z_k)])$. It is known that any cyclic group of order $n$ has a unique subgroup (which is also cyclic) of order $d$ for each divisor $d$ of $n$. Therefore, $(x_1,\dots,x_k)$ has only one one child of $col([x_i])$ which is $[(e,\dots,x_i,\dots,e)]$. Now, $col([x_i])=col([z_i])$ and $[(e,\dots,z_i,\dots,e)]$ is the only child of $[(z_1,\dots,z_k)]$ of order $col([z_i])$. Hence, by isomorphic property of $f$, we can say that $f([(e,\dots,x_i,\dots,e)])=[(e,\dots,z_i,\dots,e)]$. \end{proof}

\end{appendix}

\end{document}